\documentclass[12pt]{article}
\usepackage{graphicx}
\usepackage{amsmath}
\usepackage{amsfonts}
\usepackage{amssymb}
\usepackage{verbatim}
\providecommand{\U}[1]{\protect\rule{.1in}{.1in}}
\newtheorem{theorem}{Theorem}

\newtheorem{claim}[theorem]{Claim}

\newtheorem{definition}[theorem]{Definition}

\newtheorem{lemma}[theorem]{Lemma}

\newenvironment{proof}[1][Proof]{\noindent\textbf{#1.} }{\ \rule{0.5em}{0.5em}}

\begin{document}

\title{\textbf{A manifold of pure Gibbs states of the Ising model on the Lobachevsky
plane}}
\author{Daniel Gandolfo, Jean Ruiz, and Senya Shlosman}

\maketitle

\begin{abstract}
In this paper we construct many `new' Gibbs states of the Ising model on the
Lobachevsky plane, the \textit{millefeuilles}.
Unlike the usual states on the integer lattices, our foliated states have infinitely many interfaces.
The interfaces are rigid and fill the Lobachevsky plane with positive density.


\end{abstract}

\section{Introduction}

There is a common belief in the field of statistical mechanics that the
qualitative properties of the systems living on Cayley trees $\mathcal{T}_{n}$
and on the (tesselations $\mathcal{L}_{p,q}$ of the) Lobachevsky plane
$\mathcal{L}$ should be the same. Indeed, the behavior of the ratio
$\frac{\left\vert \partial D_{r}\right\vert }{\left\vert D_{r}\right\vert },$
where $\left\vert D_{r}\right\vert $ is the volume of the ball $D_{r}$ of
radius $r,$ and $\left\vert \partial D_{r}\right\vert $ is the volume of the
sphere $\partial D_{r}$ of the same radius, is the same for the two families
of graphs, as $r\rightarrow\infty$. Yet the models on the Cayley trees are
studied in much more details, due to the fact that there one can use the
recurrent relations, which were used in many papers, like \cite{BRZ} or
\cite{E}. On the Lobachevsky plane no such relations exists, since the graphs
$\mathcal{L}_{p,q}$ have cycles. In fact, the Cayley trees $\mathcal{T}_{n}$
are `limits' of the tesselations $\mathcal{L}_{n+1,q}$ when $q\rightarrow
\infty,$ see Fig.1. Nevertheless, some results for the Cayley trees were also
obtained for the Lobachevsky plane as well -- see \cite{B}, for example. One
such result is about the non translation invariant states of the Ising model
on the Cayley trees and Lobachevsky plane. For the Cayley trees the analogs of
Dobrushin non
translation invariant states $\left\langle \cdot\right\rangle ^{\pm},$
\cite{D}, where an interface separates $\left(  +\right)  $-phase from the
$\left(  -\right)  $-phase, were constructed by Blekher and Ganikhodzhaev in
\cite{BG}. For the Lobachevsky plane these were obtained by Series and Sinai,
\cite{SS}.

\begin{figure}[h]
\centering
\includegraphics[width=14cm]{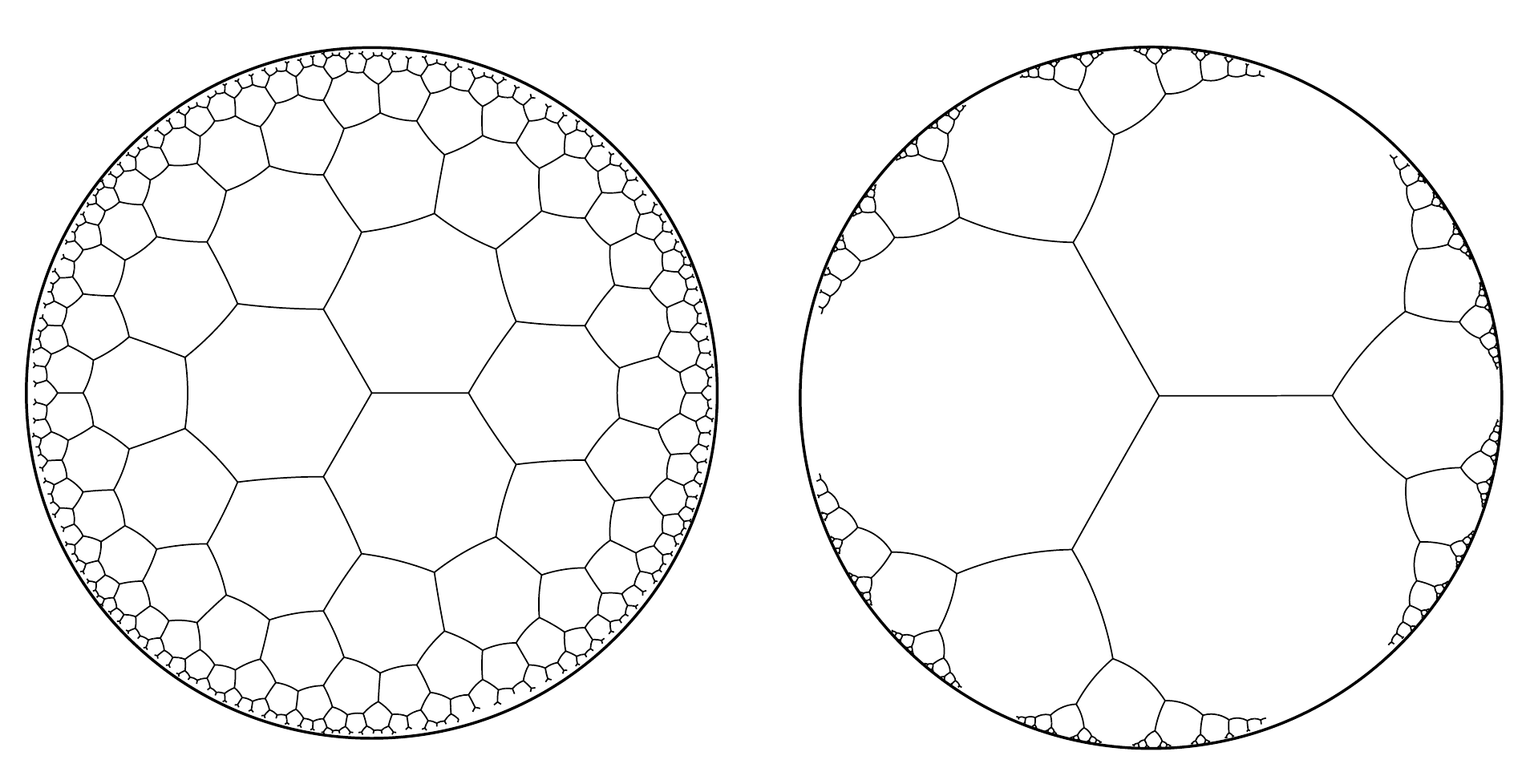}\caption{The tesselation
$\mathcal{L}_{3,7}$ and the Cayley tree $\mathcal{T}_{2}$.}
\label{7_3}
\end{figure}

Later it was discovered that there are much more non translation invariant
states on Cayley trees $\mathcal{T}_{n}$ than there are those for the case of
the lattice $\mathbb{Z}^{\nu}.$ Namely, in \cite{RR} the authors constructed
Gibbs states of the ferromagnetic Ising model on $\mathcal{T}_{n}$, $n\geq4,$
which they called `weakly periodic'. For such a state $\left\langle
\cdot\right\rangle $ the expectation $\left\langle \sigma_{t}\right\rangle $
of the spin $\sigma_{t}$ at $t\in\mathcal{T}_{n}$ is a `periodic' function on
$\mathcal{T}_{n},$ taking finitely many different values. Thus, it is
different both from the $\left(  +\right)  $-state $\left\langle
\cdot\right\rangle ^{+}$ and the $\left(  -\right)  $-state $\left\langle
\cdot\right\rangle ^{-}$, as well as from the Dobrushin $\left\langle
\cdot\right\rangle ^{\pm}$ states. Indeed, for every $t$ we have $\left\langle
\sigma_{t}\right\rangle ^{+}\equiv m^{\ast},$ $\left\langle \sigma
_{t}\right\rangle ^{-}\equiv-m^{\ast},$ while the function $\left\langle
\sigma_{t}\right\rangle ^{\pm}$ takes infinitely many values -- provided the
temperature $\beta^{-1}$ is low enough. Here $m^{\ast}=m^{\ast}\left(
\beta\right)  $ is the spontaneous magnetization.

Then in \cite{GRS} we have given a general construction of a grand family of
pure states on Cayley trees $\mathcal{T}_{n},$ again for $n\geq4.$ One can
think of these states as having infinitely many rigid $\pm$-interfaces, and
moreover the density of those interfaces is positive. (The reader should not
think that the trees $\mathcal{T}_{n}$ with $n\geq4$ have properties which the
lighter trees $\mathcal{T}_{2}$ and $\mathcal{T}_{3}$ do not enjoy. Indeed, in
the Section 4 of the present paper we explain how to modify the constructions
of \cite{GRS} in order to include all the Cayley trees $\mathcal{T}_{n}$,
$n\geq2.$ So the trees $\mathcal{T}_{n}$ with $n\geq4$ are not special in that respect.)

The purpose of the present paper is to construct similar family of states on
the tesselations $\mathcal{L}_{p,q}$ of the Lobachevsky plane $\mathcal{L}.$
Here $\mathcal{L}_{p,q}$ is the planar graph, such that each face has $p$
edges, every vertex is incident to $q$ edges, and $1/q+1/p<1/2$. Every such
graph can be isometrically embedded into $\mathcal{L}.$ The main result of the
present paper is the construction of Gibbs states of the Ising model on
$\mathcal{L}_{p,q},$ which have infinitely many Dobrushin $\pm$-interfaces. In
our case the interfaces are 1D random curves. They are rigid, as is the
interface of the Dobrushin state $\left\langle \cdot\right\rangle ^{\pm}$ of
the Ising model on $\mathbb{Z}^{3}.$ But unlike the 3D lattice, when there can
be at most one such interface, on the Lobachevsky plane $\mathcal{L}$ one can
have infinitely many of them. Moreover, they can have positive density. We
call these states the \textbf{Foliated States}, or just the millefeuilles.

In the next section we introduce special families of geodesics, which we call
geodesical families. We provide examples of such families and establish their
properties. The geodesical families define foliated\textbf{ }ground state
configurations. We prove in the Section 3 that these ground state
configurations are stable: the thermal fluctuations do not destroy them. In
other words, there are low-temperature Gibbs states, which are small
perturbations of the foliated\textbf{ }ground states. In the Section 4 we study
the foliated states on the Cayley trees.

\section{Geodesics and geodesical families}

Let $M$ be a locally compact Riemannian manifold. A (finite or countable)
family $\Gamma$ of smooth curves $\gamma_{i}\in M,$ either closed or coming
from and going to infinity, will be called a \textbf{geodesical family}, if no
bounded surgery can decrease its length. By this we mean the following: let
$\mathbf{s}=\left\{  s_{i}=\left[  \alpha_{2i-1},\alpha_{2i}\right]
\subset\gamma_{i},i=1,...,k\right\}  $ be a collection of segments of the
curves $\gamma_{i},$ $i=1,...,k.$ We want to remove the segments $s_{i}$ and
to interconnect the remaining family of curves, $\left\{  \gamma
_{i}\smallsetminus s_{i},i=1,...,k\right\}  $ in a different way. To do this,
let us consider a permutation $\pi$ of the set $\left\{  \alpha_{1}%
,...,\alpha_{2k}\right\}  ,$ and let $\mathbf{\bar{s}}=\left\{  \bar{s}%
_{i}=\left[  \pi\left(  \alpha_{2i-1}\right)  ,\pi\left(  \alpha_{2i}\right)
\right]  \right\}  $ be some collection of continuous curves, connecting the
points $\pi\left(  \alpha_{2i-1}\right)  $ and $\pi\left(  \alpha_{2i}\right)
,$ $i=1,...,k.$ The length increment
\[
\Delta\left(  \mathbf{s},\mathbf{\bar{s}}\right)  =\sum_{i=1}^{k}\left\vert
\bar{s}_{i}\right\vert -\sum_{i=1}^{k}\left\vert s_{i}\right\vert
\]
is well-defined. The family $\left\{  \gamma_{i}\right\}  $ is called
geodesical, if for any $k\geq1,$ any collection $\mathbf{s,}$ and permutation
$\pi$ and any collection $\mathbf{\bar{s}}\neq\mathbf{s}$ the increment
$\Delta\left(  \mathbf{s},\mathbf{\bar{s}}\right)  $ is strictly positive. In
particular, each curve $\gamma_{i}$ from a geodesical family has to be a geodesic.

(A reader with background from stat. mechanics recognizes immediately that our
definition is inspired by that of a ground state configuration. The
generalization to the case of minimal surfaces is immediate, but we will not
need it this paper.)

For example, if $\gamma_{1}$ and $\gamma_{2}$ are two geodesics, and
$\gamma_{1}\cap\gamma_{2}\neq\varnothing,$ then the pair $\gamma_{1}%
,\gamma_{2}$ is never a geodesical family. If $M$ is $\mathbb{R}^{2}$ or
$\mathbb{R}^{n},$ $n>2,$ then any geodesical family contains at most one
geodesic, which is a straight line.

The situation is different when $M$ is the Lobachevsky plane, $\mathcal{L}$.
Let $\gamma_{1},\gamma_{2}$ be two geodesics, and $x_{i}^{\prime}%
,x_{i}^{\prime\prime}$ be their end-points at the absolute, $i=1,2.$ The
cross-ratio of the two pairs of points $\left(  x_{1}^{\prime},x_{1}%
^{\prime\prime}\right)  $ and $\left(  x_{2}^{\prime},x_{2}^{\prime\prime
}\right)  $ is defined by%
\[
\bar{R}\left(  x_{1}^{\prime},x_{1}^{\prime\prime};x_{2}^{\prime}%
,x_{2}^{\prime\prime}\right)  =\frac{\left(  x_{2}^{\prime}-x_{1}^{\prime
}\right)  \left(  x_{2}^{\prime\prime}-x_{1}^{\prime\prime}\right)  }{\left(
x_{2}^{\prime}-x_{1}^{\prime\prime}\right)  \left(  x_{2}^{\prime\prime}%
-x_{1}^{\prime}\right)  }.
\]
For us it will be more convenient to use another version of it, which is given
by
\[
R\left(  x_{1}^{\prime},x_{1}^{\prime\prime};x_{2}^{\prime},x_{2}%
^{\prime\prime}\right)  =-\frac{\left(  x_{1}^{\prime\prime}-x_{1}^{\prime
}\right)  \left(  x_{2}^{\prime\prime}-x_{2}^{\prime}\right)  }{\left(
x_{1}^{\prime\prime}-x_{2}^{\prime}\right)  \left(  x_{2}^{\prime\prime}%
-x_{1}^{\prime}\right)  }
=\bar{R}-1.
\]
Then the pair $\gamma_{1},\gamma_{2}$ of non-crossing geodesics is a
geodesical family, iff the cross-ratio $R\left(  x_{1}^{\prime},x_{1}%
^{\prime\prime};x_{2}^{\prime},x_{2}^{\prime\prime}\right)  $ of the quadruple
$x_{1}^{\prime},x_{1}^{\prime\prime},x_{2}^{\prime},x_{2}^{\prime\prime}$ is
less than $1$.

In what follows we will use for the quantity $R\left(  x_{1}^{\prime}%
,x_{1}^{\prime\prime},x_{2}^{\prime},x_{2}^{\prime\prime}\right)  $ the
notation $R\left(  \gamma_{1},\gamma_{2}\right)  .$ For later use we need to
choose a scale length on $\mathcal{L}.$ We do it by imposing the condition
that for every pair $\gamma_{1},\gamma_{2}$ of non-intersecting geodesics with
the cross-ratio of the quadruple $x_{1}^{\prime},x_{1}^{\prime\prime}%
,x_{2}^{\prime},x_{2}^{\prime\prime}$ equal to $1,$ we have $\mathrm{dist}%
\left(  \gamma_{1},\gamma_{2}\right)  =1.$ When $R\left(  \gamma_{1}%
,\gamma_{2}\right)  \rightarrow0,$ we have that $\mathrm{dist}\left(
\gamma_{1},\gamma_{2}\right)  \rightarrow\infty.$

We are going to present a countable geodesical family $\Gamma,$ having\textit{
positive density. }That means that there exists a value $R>0,$ such that for
every point $x\in\mathcal{L}$ the disc $D_{x}\left(  R\right)  ,$ centered at
$x$ and having radius $R$ intersects some of the curves from $\Gamma.$ The
construction is the following:

\textbf{Construction 1. } Let $0\in\mathcal{L}$ be an arbitrary point, which
we will fix and will call an origin. Our construction will be inductive, and
will start `near' $0$ and will proceed away from $0$ to infinity.

\begin{figure}[!ht]
\centering
{\includegraphics[width=13cm]{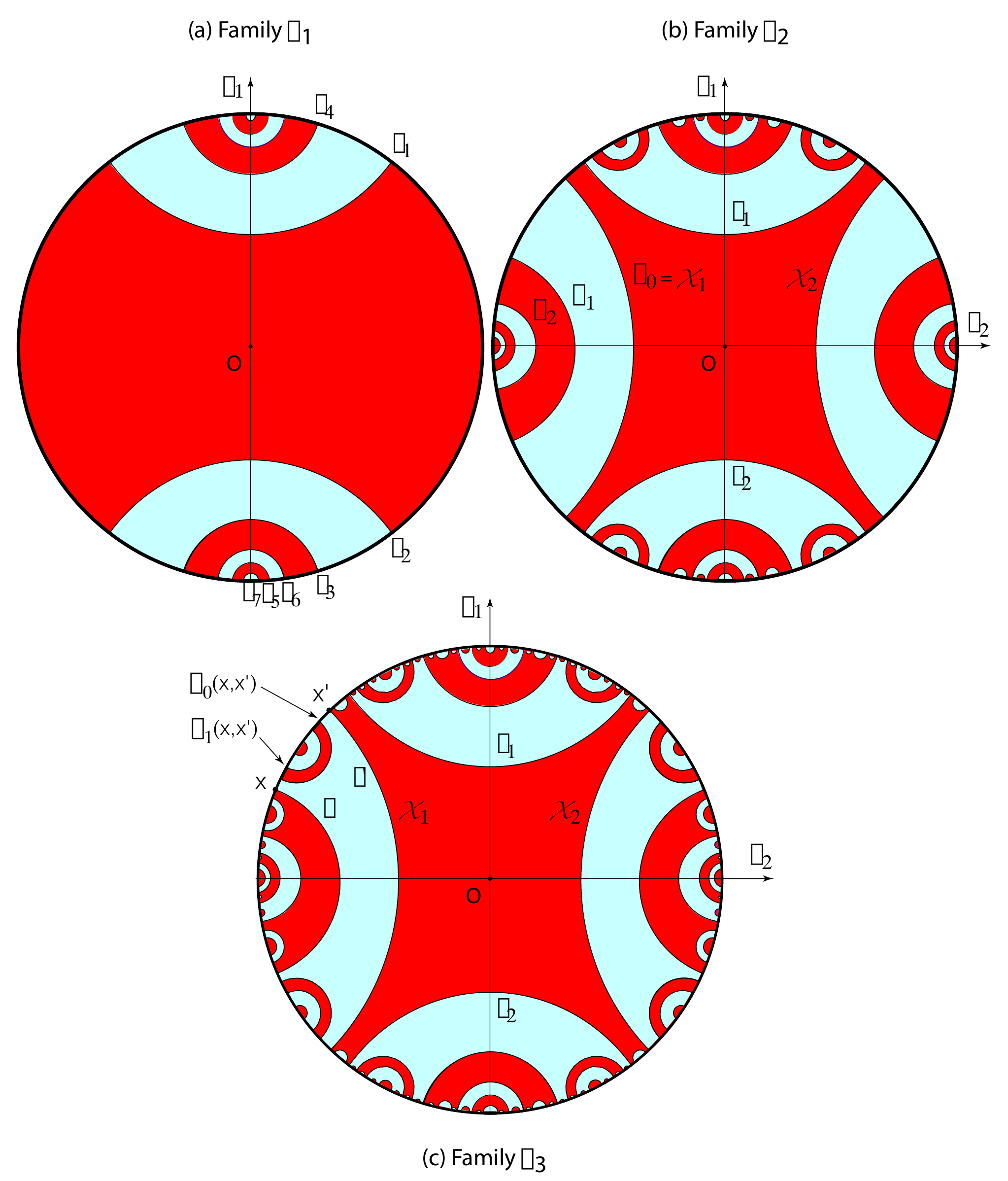}
\label{Constr1}}
\caption{Construction 1, the first three steps.}%
\end{figure}

Let us fix a small $\alpha<1.$ Let $\gamma_{1},\gamma_{2}$ be two geodesics,
with $\gamma_{1}$ being centrally symmetric to $\gamma_{2}$ with respect to
$0$. Additionally we suppose that for their
end-points $x_{1}^{\prime} ,x_{1}^{\prime\prime},x_{2}^{\prime},x_{2}
^{\prime\prime}$ at the absolute we have $R\left(  x_{1}^{\prime}
,x_{1}^{\prime\prime},x_{2}^{\prime} ,x_{2}^{\prime\prime}\right)  =\alpha.$
Let $\gamma_{3}$ be an axial reflection of $\gamma_{1}$ in $\gamma_{2},$
$\gamma_{4}$ -- a reflection of $\gamma_{2}$ in $\gamma_{1}.$ Let $\gamma
_{5},$ $\gamma_{6},$ $\gamma_{7}$ be reflections of $\gamma_{1},$ $\gamma
_{2},$ $\gamma_{4}$ in $\gamma_{3},$ and so on. Proceeding in this way, we
obtain an infinite sequence of geodesics, which we denote by $\Gamma_{1}$, see Fig.\ref{Constr1}(a).
Note that if the two geodesics $\gamma,\gamma^{\prime}\in\Gamma_{1}$ are
neighbors, then the cross-ratio $R\left(  \gamma,\gamma^{\prime}\right)
=\alpha.$ For every geodesic $\gamma_{i}$ in $\Gamma_{1}$ denote by $c\left(
\gamma_{i}\right)  $ the point on $\gamma_{i}$ which is the closest to $0.$
Clearly, there exists a geodesic $\delta_{1}$, passing through $0$ and
containing all the points $c\left(  \gamma_{i}\right)  .$ Let $\delta_{2}$ be
a geodesic passing through $0$ and orthogonal to $\delta_{1}.$ Clearly,
$\gamma_{2}$ is a reflection of $\gamma_{1}$ in $\delta_{2}.$

Let us now fill in the spaces between every two consecutive geodesics
$\gamma,\gamma^{\prime}\in\Gamma_{1}.$ Because of the above symmetries it is
sufficient to fill in the strip $S\left(  \gamma_{1},\gamma_{2}\right)  $
between $\gamma_{1}$ and $\gamma_{2}$; reflecting this `filling' in curves of
$\Gamma_{1}$ we get it for all neighboring curves in $\Gamma_{1}.$ We do it as
follows. Let $\varkappa_{1}$ and $\varkappa_{2}$ be two non-intersecting
geodesics in $S\left(  \gamma_{1},\gamma_{2}\right)  ,$ having the property:
$R\left(  \gamma_{i},\varkappa_{j}\right)  =\alpha.$ The existence and
uniqueness of this pair is straightforward. Note that $\varkappa_{2}$ is a
reflection of $\varkappa_{1}$ in $\delta_{1},$ and that the centers $c\left(
\varkappa_{j}\right)  $ belong to $\delta_{2}.$ Let us fill in the `inside' of
$\varkappa_{1}$ by the sequence of geodesics $\chi_{i},$ $i=1,2,...,$
$\chi_{0}\equiv\varkappa_{1},$ which are defined by the properties: $R\left(
\chi_{i},\chi_{i+1}\right)  =\alpha,$ $c\left(  \chi_{i}\right)  \in\delta
_{2},$ and $\mathrm{dist}\left(  c\left(  \chi_{i}\right)  ,0\right)
\nearrow\infty,$ $i=0,1,2,....$ We do the same for $\varkappa_{2}.$ Applying
the reflection symmetries we fill in all the strips between the neighboring
curves in $\Gamma_{1}.$ The obtained family, together with curves in
$\Gamma_{1},$ is denoted by $\Gamma_{2}$, see Fig.\ref{Constr1}(b).

Note that for every two neighboring curves $\gamma,\gamma^{\prime}\in
\Gamma_{2}$ we have $R\left(  \gamma,\gamma^{\prime}\right)  =\alpha.$ The
notion `neighboring' means that one can choose ends $x,x^{\prime}$ for
$\gamma,\gamma^{\prime}$ in such a way that between them there are no ends of
other members of $\Gamma_{2}.$ For each such pair of ends $x,x^{\prime}$ of
the curves $\gamma,\gamma^{\prime}$ we reproduce the above construction,
obtaining the sequence $\chi_{i}\left(  x,x^{\prime}\right)  ,$
$i=0,1,2,...\ $of geodesics with ends between $x$ and $x^{\prime}.$ In
particular, $R\left(  \chi_{i}\left(  x,x^{\prime}\right)  ,\chi_{i+1}\left(
x,x^{\prime}\right)  \right)  =\alpha=$ $R\left(  \chi_{0}\left(  x,x^{\prime
}\right)  ,\gamma\right)  =$ $R\left(  \chi_{0}\left(  x,x^{\prime}\right)
,\gamma^{\prime}\right)  .$ The resulting family is denoted by $\Gamma_{3}$,
see Fig.\ref{Constr1}(c).

The last step can be iterated, so we inductively can define the families
$\Gamma_{1}\subset\Gamma_{2}\subset...;$ the final result $\Gamma_{\infty
}=\cup_{k}\Gamma_{k}$ is the desired collection. $\blacksquare$

The Millefeuille Construction 1 above results in a `ground state' with `zero
mean magnetization'. We present now a generalization, which will have a
non-zero magnetization.

\bigskip

The \textbf{Construction 2} will be defined by the two parameters: $\alpha$ and $\eta.$
The construction of the family $\Gamma_{\infty}=\Gamma_{\infty}\left(
\alpha,\eta\right)  $ is inductive. The data after the completion of the
$k$-th step consists from

\begin{itemize}
\item the family of $k$ non-intersecting geodesics $\gamma_{i},$ $i=1,...,k,$

\item a choice of one arc among the $2k$ non-intersecting arcs on the
absolute, which are defined by the $2k$ end-points of the curves $\gamma_{i},$

\item the chess-board assignment of the $+$ or $-$ signs ($\equiv$phases) to
each of the $k+1$ regions of the plane $\mathcal{L}.$
\end{itemize}

Let $\left(  x,y\right)  $ be the selected arc, and $\left(  y,z\right)  $ be
the next one (clockwise). Then after the $k+1$-th step one extra geodesic
$\gamma_{k+1}$ is added, with end-points inside the arc $\left(  x,y\right)
,$ while the arc $\left(  y,z\right)  $ is declared as chosen. Now we have to
describe the rule of constructing the curve $\gamma_{k+1}$ given all the
previous curves $\gamma_{i},$ the new assignment of signs, plus we need to
specify the first step of induction.

The sign assignment is simple. Let $\mathcal{L}_{0},\mathcal{L}_{1}%
,...,\mathcal{L}_{k}$ are connected components of $\mathcal{L},$ which are
defined by the lines $\gamma_{1},...,\gamma_{k}.$ Each of them already has its
sign. The geodesics $\gamma_{k+1}$ belongs to one of this components, say
$\mathcal{L}_{j},$ and it splits $\mathcal{L}_{j}$ into two connected
components, one of which borders some of the previous curves $\gamma_{i}.$ Let
us call this component exterior to $\gamma_{k+1}$, while other one will be
called interior, and will be denoted by $\mathcal{L}_{k+1}.$ Let the `new'
component $\mathcal{L}_{j}\equiv\mathcal{L}_{j}^{"new"}=\mathcal{L}%
_{j}^{"old"}\smallsetminus\mathcal{L}_{k+1}.$ Then the `new' component
$\mathcal{L}_{j}$ retains the sign of the `old' one, while the component
$\mathcal{L}_{k+1}$ gets the opposite one. All other components retain their signs.

In the beginning we pick a point $0$ on the plane, which will be called the
origin. We take for $\gamma_{1}$ an arbitrary (maximal) geodesic, passing
through $0,$ and we choose arbitrarily one of the two arcs (semicircles) on
the absolute, which are defined by the endpoints of $\gamma_{1}.$ The
resulting two components $\mathcal{L}_{0}$ and $\mathcal{L}_{1}$ of the plane
$\mathcal{L}$ get opposite signs $+$ and $-.$

\begin{figure}[h]
\centering
\includegraphics[width=8cm]{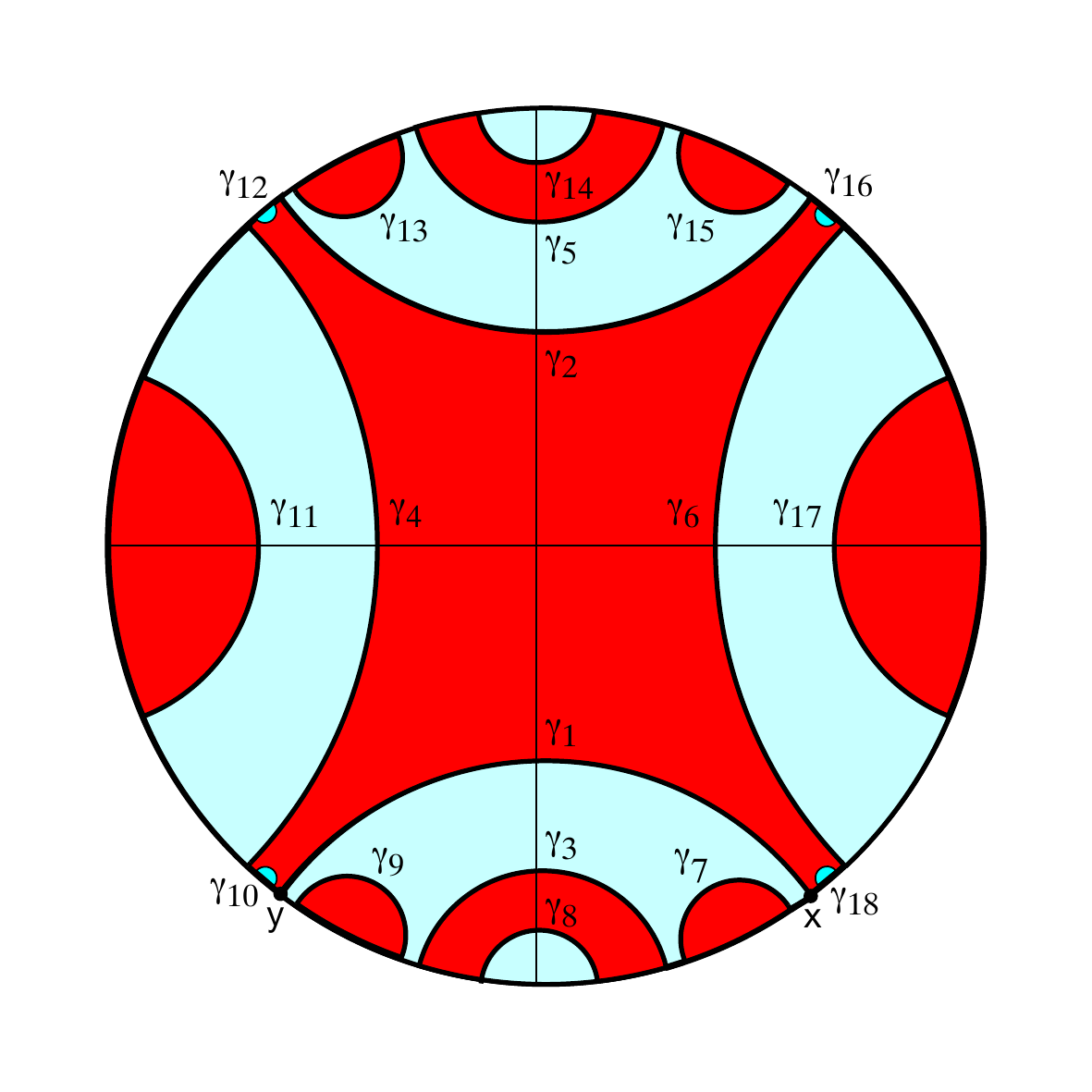}\caption{Construction 2, the first 18
geodesics.}%
\label{Constr2}%
\end{figure}

The construction of the curve $\gamma_{k+1}$ is defined by one of our extra
real parameters $\alpha,\eta>0;$ we use $\alpha$ if the curve $\gamma_{k+1}$
traverses the `$\left(  +\right)  $-phase', and $\eta$ if it goes through
`$\left(  -\right)  $-phase'. The construction is the same in both cases, so
we consider only the case of the `$\left(  +\right)  $-phase', where we use
the parameter $\alpha.$ So let $\left(  x,y\right)  $ is our chosen arc on the
absolute, and we want to pick two points $\left(  u,v\right)  $ on it, which
uniquely define the geodesic $\gamma_{k+1}=\gamma\left(  u,v\right)  ,$
joining them, see Fig.\ref{Constr2}. Suppose first that the two points $x,y$
are endpoints of the same geodesic $\gamma_{i},$ constructed earlier. Then the
pair $u,v$ is uniquely defined by the two properties:

\begin{itemize}
\item the cross-ratio $\left(  u,v;y,x\right)  =\alpha,$

\item the geodesic, which passes through $0$ and is perpendicular to
$\gamma_{i},$ is also perpendicular to $\gamma\left(  u,v\right)  .$
\end{itemize}

In the remaining case the point $x$ is an end-point of a geodesic $\gamma
_{i}=\gamma\left(  t,x\right)  ,$ while $y$ is an end-point of a geodesic
$\gamma_{j}=\gamma\left(  y,z\right)  .$ Then the pair $u,v$ is uniquely
defined by the two properties:

\begin{itemize}
\item the cross-ratio $\left(  u,v;y,z\right)  =\alpha,$

\item the cross-ratio $\left(  u,v;t,x\right)  =\alpha.$ $\blacksquare$
\end{itemize}

If $\alpha=\eta,$ the Construction 2 gives the same geodesical family as
Construction 1.

\bigskip

Let us check that indeed the above families of geodesics are geodesical
families. We will represent the Lobachevsky plane $\mathcal{L}$ as a disc
$D_{1}$ of unit radius in $\mathbb{C}^{1},$ centered at the origin $0;$ then
the absolute $\mathcal{L}^{\infty}$ will be represented by the unit circle
$C_{1}$. Let $z_{1},z_{2},z_{3},z_{4}\in\mathcal{L}^{\infty}$ be four points,
going clockwise. Note that when the points $z_{1},z_{2},z_{3},z_{4}$ are made
by an isometry of $\mathcal{L}$ into a rectangle, and the cross-ratio
$R\left(  z_{1},z_{2};z_{3},z_{4}\right)  $ is small, the side $\left[
z_{1},z_{2}\right]  $ of this rectangle is much shorter than $\left[
z_{2},z_{3}\right].$

Let $\gamma_{1},\gamma_{2},...,\gamma_{k}\subset\mathcal{L}$ be a family of
non-intersecting doubly-infinite geodesics on Lobachevsky plane $\mathcal{L}$.
For two points $x^{\prime},x^{\prime\prime}$ on the absolute $\mathcal{L}%
^{\infty}$ we denote by $\gamma\left(  x^{\prime},x^{\prime\prime}\right)  $
the geodesic connecting them. Let $\gamma_{i}=\gamma\left(  x_{i}^{\prime
},x_{i}^{\prime\prime}\right)  ,$ $i=1,2,...,k.$ The following lemma claims
that if all the geodesics $\gamma_{1},\gamma_{2},...,\gamma_{k}$ are far away
from each other, then they form a geodesical family.

\begin{lemma}
\label{family} Suppose that for some $\alpha$ small enough the geodesics
$\gamma_{i}=\gamma\left(  x_{i}^{\prime},x_{i}^{\prime\prime}\right)  ,$
$i=1,...,k$ have the property that for any $i\neq j$%
\[
R\left(  x_{i}^{\prime},x_{i}^{\prime\prime};x_{j}^{\prime},x_{j}%
^{\prime\prime}\right)  <\alpha.
\]
Then the family $\gamma_{1},\gamma_{2},...,\gamma_{k}$ is geodesical.
\end{lemma}

\begin{proof}
Note that the condition $R\left(  x_{i}^{\prime},x_{i}^{\prime\prime}%
;x_{j}^{\prime},x_{j}^{\prime\prime}\right)  <\alpha$ implies that the
Lobachevsky distance $\lambda$ satisfies
\[
\lambda\left(  \gamma_{i},\gamma_{j}\right)  \equiv\inf_{y_{i}\in\gamma
_{i},y_{j}\in\gamma_{j}}\lambda\left(  y_{i},y_{j}\right)  >L\left(
\alpha\right)  ,
\]
with $L\left(  \alpha\right)  \rightarrow\infty$ as $\alpha\rightarrow0.$ Let
$\delta_{ij}$ be the common perpendicular to $\gamma_{i}$ and $\gamma_{j},$
and $\Delta_{ij}\in\gamma_{i},$ $\Delta_{ji}\in\gamma_{j}$ be the feet of this
perpendicular. We have that the length $\left\vert \delta_{ij}\right\vert
=\lambda\left(  \Delta_{ij},\Delta_{ji}\right)  \equiv\lambda\left(
\gamma_{i},\gamma_{j}\right)  >L\left(  \alpha\right)  .$ Define $L=\max
_{ij}\left\vert \delta_{ij}\right\vert .$

For the future use we will compute the quantity $L\left(  \alpha\right)  .$
Let $x>0$ be large, and $\gamma_{1}\left(  -x-1,-x+1\right)  ,\gamma
_{2}\left(  x-1,x+1\right)  $ be two unit semicircles in the upper
half--plane, centered at $-x$ and $x,$ i.e. geodesics in the upper half--plane
model. Then the cross-ratio $R\left(  -x-1,-x+1;x-1,x+1\right)  =\frac
{1}{x^{2}-1},$ while the distance $\mathrm{dist}\left(  \gamma_{1},\gamma
_{2}\right)  \cong\mathrm{dist}\left(  \left(  -x,1\right)  ,\left(
x,1\right)  \right)  =\cosh^{-1}\left(  1+2x^{2}\right)  \cong\ln x^{2},$ so
\begin{equation}
L\left(  \alpha\right)  \cong\ln\alpha^{-1}. \label{36}%
\end{equation}

Let us start with the case of two geodesics, $\gamma_{1}=\gamma_{1}\left(
x_{1},x_{2}\right)  $ and $\gamma_{2}=\gamma_{2}\left(  x_{3},x_{4}\right)  .$
We can suppose that $\gamma_{2}$ is symmetric to $\gamma_{1}$ with respect to
the origin $0\in\mathcal{L}.$ The common perpendicular $\delta_{12}$ passes
through $0$, and $\lambda\left(  0,\Delta_{12}\right)  =\lambda\left(
0,\Delta_{21}\right)  .$ Consider now the surgery, which is almost the same as
the pair of geodesics $\varkappa_{1}=\varkappa_{1}\left(  x_{1},x_{4}\right)
$ and $\varkappa_{2}=\varkappa_{2}\left(  x_{2},x_{3}\right)  .$ We want to
show that $\left\vert \varkappa_{1}\right\vert +\left\vert \varkappa
_{2}\right\vert \gg\left\vert \gamma_{1}\right\vert +\left\vert \gamma
_{2}\right\vert $ (renormalized in the obvious way), provided $L\left(
\alpha\right)  $ is large.

For that let us consider the triangle $0,\Delta_{12},A,$ where $A\in\gamma
_{1}$ with $\lambda\left(  A,0\right)  >L\left(  \alpha\right)  ,$ and $A$ is
on the same side from $0$ as $x_{1}.$ Note that both the distances
$\lambda\left(  0,\varkappa_{1}\right)  $ and $\lambda\left(  A,\varkappa
_{1}\right)  $ become small for small $\alpha.$ Now we will use the following

\textbf{Estimate} $1_{\mathbb{H}}$ from \cite{SS}, p. 69: Let $XYZ$ be a
triangle on $\mathcal{L},$ made by three geodesics, with the angle
$\measuredangle XYZ=\frac{\pi}{2}.$ Then%
\begin{equation}
\left\vert XZ\right\vert \geq\left\vert XY\right\vert +\left\vert
YZ\right\vert -c, \label{31}%
\end{equation}
for some universal constant $c$, independent of $X,$ $Y$ and $Z.$

It follows from $\left(  \ref{31}\right)  $ that $\left\vert 0A\right\vert
\geq\frac{1}{2}\delta_{12}+\left\vert \Delta_{12}A\right\vert -c,$ which
implies that $\left\vert \varkappa_{1}\right\vert +\left\vert \varkappa
_{2}\right\vert \geq\left\vert \gamma_{1}\right\vert +\left\vert \gamma
_{2}\right\vert +2\delta_{12}-4c-\varepsilon\left(  \alpha\right)  ,$ with
$\varepsilon\left(  \alpha\right)  \rightarrow0$ as $\alpha\rightarrow0.$
Since $\delta_{12}\rightarrow\infty$ as $\alpha\rightarrow0,$ that proves our claim.

Consider now the general case, when we have $2k$ points at the absolute,
$x_{1}^{\prime},x_{1}^{\prime\prime},x_{2}^{\prime},...,x_{k}^{\prime\prime},$
going clockwise. Note that the only permutation $\bar{\pi}$ we have to
consider is given by
\[
\bar{\pi}\left(  x_{1}^{\prime},x_{1}^{\prime\prime},x_{2}^{\prime
},...,x_{k-1}^{\prime\prime},x_{k}^{\prime},x_{k}^{\prime\prime}\right)
=x_{1}^{\prime},x_{k}^{\prime\prime},x_{k}^{\prime},x_{k-1}^{\prime\prime
},...,x_{2}^{\prime},x_{1}^{\prime\prime}.
\]
First of all, we do not have to consider permutations, for which the geodesics
$\gamma\left(  \pi\left(  x_{i}^{\prime}\right)  ,\pi\left(  x_{i}%
^{\prime\prime}\right)  \right)  $ and $\gamma\left(  \pi\left(  x_{j}%
^{\prime}\right)  ,\pi\left(  x_{j}^{\prime\prime}\right)  \right)  $
intersect for some $i\neq j.$ Indeed, let $\tau$ be a transposition,
exchanging the points $\pi\left(  x_{i}^{\prime\prime}\right)  $ and
$\pi\left(  x_{j}^{\prime\prime}\right)  ;$ it is easy to see that
\begin{multline*}
\left\vert \gamma\left(  \pi\left(  x_{i}^{\prime}\right)  ,\pi\left(
x_{i}^{\prime\prime}\right)  \right)  \right\vert +\left\vert \gamma\left(
\pi\left(  x_{j}^{\prime}\right)  ,\pi\left(  x_{j}^{\prime\prime}\right)
\right)  \right\vert >\\
\left\vert \gamma\left(  \tau\pi\left(  x_{i}^{\prime}\right)  ,\tau\pi\left(
x_{i}^{\prime\prime}\right)  \right)  \right\vert +\left\vert \gamma\left(
\tau\pi\left(  x_{j}^{\prime}\right)  ,\tau\pi\left(  x_{j}^{\prime\prime
}\right)  \right)  \right\vert .
\end{multline*}
Next, for every remaining permutation $\pi$ consider the union of all the
curves $\gamma\left(  x_{i}^{\prime},x_{i}^{\prime\prime}\right)  $ and all
the curves $\gamma\left(  \pi\left(  x_{i}^{\prime}\right)  ,\pi\left(
x_{i}^{\prime\prime}\right)  \right)  .$ This union consists of several --
say, $l$ -- connected components, $l\leq k.$ If $l>1,$ then we can consider
each component separately, thus reducing the number $2k$ of points
$x_{1}^{\prime},x_{1}^{\prime\prime},x_{2}^{\prime},...,x_{k}^{\prime\prime}$
to smaller values and use induction on $k.$

\begin{figure}[h]
\centering
\includegraphics[width=14cm]{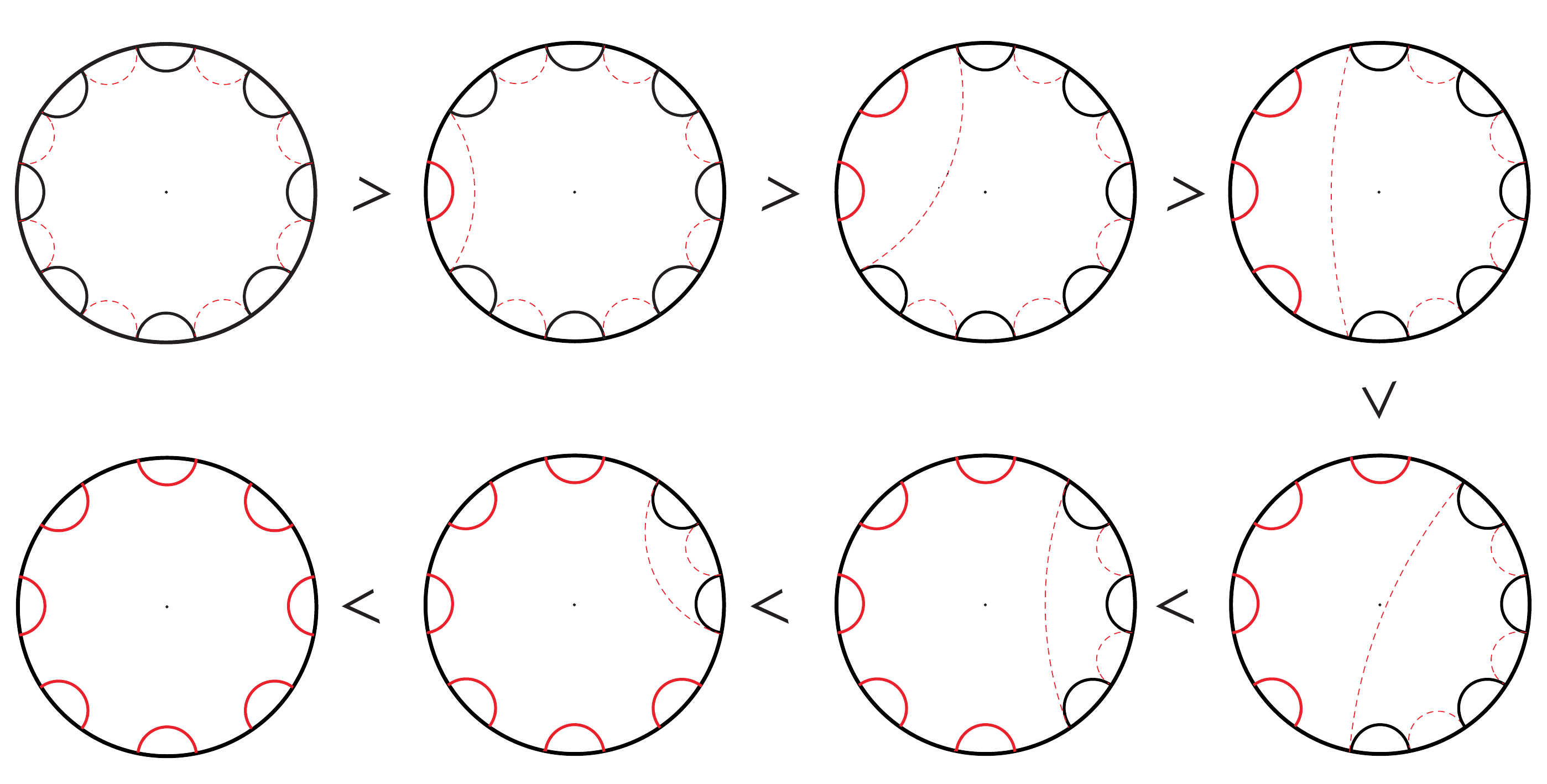}\caption{The length decrease
through surgeries.}%
\label{majoration}%
\end{figure}

We will compare the permutation $\bar{\pi}$ with another one, $\bar{\pi
}^{\prime},$ given by

\noindent$\bar{\pi}^{\prime}\left(  x_{1}^{\prime},x_{1}^{\prime\prime}%
,x_{2}^{\prime},x_{2}^{\prime\prime},...,x_{k-1}^{\prime\prime},x_{k}^{\prime
},x_{k}^{\prime\prime}\right)  =x_{1}^{\prime\prime},x_{1}^{\prime}%
,x_{2}^{\prime},x_{k}^{\prime\prime},x_{k}^{\prime},x_{k-1}^{\prime\prime
},...,x_{2}^{\prime\prime}.$ This one has two components: the first one
contains the two points $x_{1}^{\prime},x_{1}^{\prime\prime},$ while the
second one incorporates all the rest, see Fig.\ref{majoration}. We will show
that by doing surgery $\bar{\pi}\rightsquigarrow\bar{\pi}^{\prime}$ just
described -- which makes two cycles from one -- we diminish the total length.
The argument is the same for all $k,$ so we will consider the case $k=3.$ The
computations are simpler in the half--plane model, in which case the absolute
is just the real line $\mathbb{R}^{1}$.

Without loss of generality we can take two pairs of points to be
$-x-1,-x+1;x-1,x+1$, see Fig.\ref{H2}. Their cross-ratio is%
\[
R\left(  -x-1,-x+1;x-1,x+1\right)  =\frac{1}{x^{2}-1},
\]
and our assumption that the two geodesics $\gamma\left(  -x-1,-x+1\right)
,\gamma\left(  x-1,x+1\right)  $ are far away is satisfied once $x$ is large
enough. Let there be another pair, $y-z,y+z,$ lying in between,
$-x+1<y-z,y+z<x-1,$ and we want the two cross-ratios to be as small:
\[
\ \ R\left(  -x-1,-x+1;y-z,y+z\right)  =\frac{4z}{\left(  x+y\right)
^{2}-\left(  z+1\right)  ^{2}}\leq\frac{1}{x^{2}-1}.
\]
\begin{align*}
R\left(  y-z,y+z;x-1,x+1\right)   &  =\frac{4z}{\left[  x-1-y-z\right]
\left[  x+1-y+z\right]  }\\
&  \equiv\frac{4z}{\left(  x-y\right)  ^{2}-\left(  z+1\right)  ^{2}}\leq
\frac{1}{x^{2}-1}.
\end{align*}

\begin{figure}[h]
\centering
\includegraphics[width=13cm]{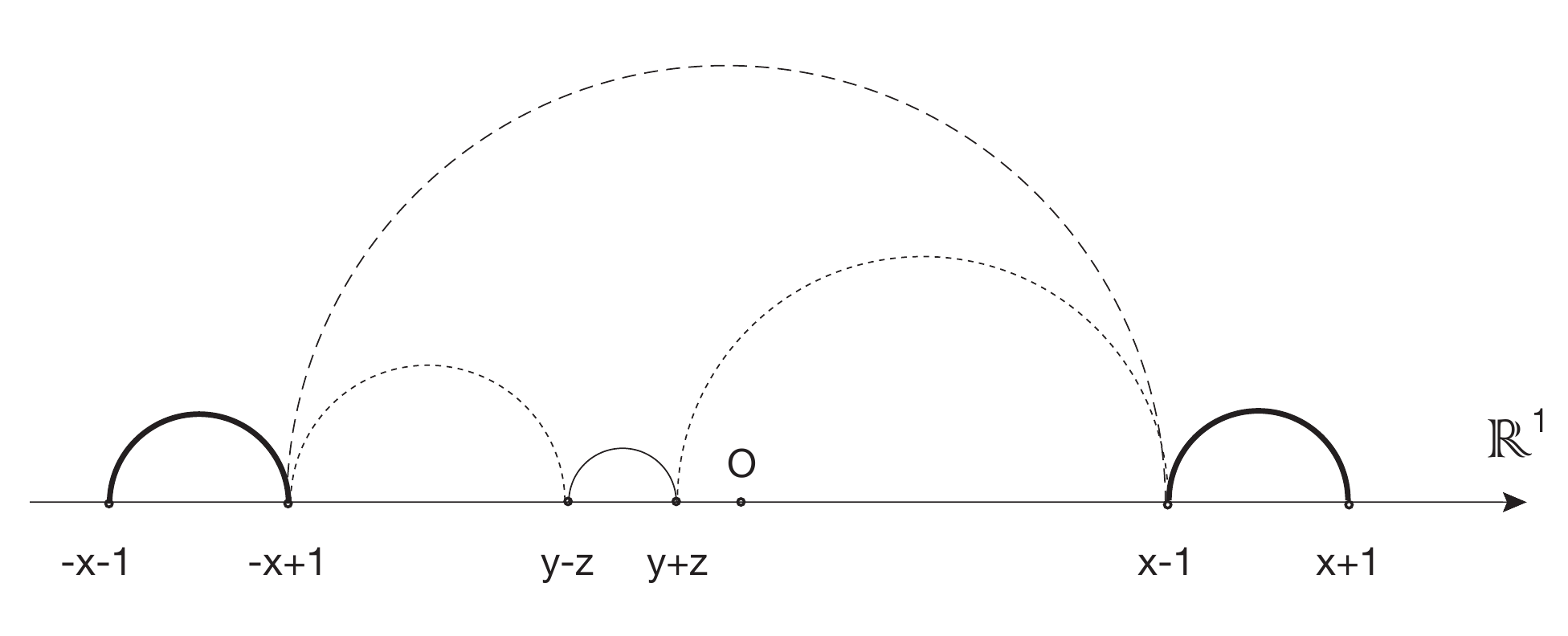}\caption{The first surgery step in the
half--plane model.}%
\label{H2}%
\end{figure}

Our goal is to show that under our assumptions we have
\[
\left\vert \gamma\left(  -x+1,y-z\right)  \right\vert +\left\vert
\gamma\left(  y+z,x-1\right)  \right\vert \gg\left\vert \gamma\left(
y-z,y+z\right)  \right\vert +\left\vert \gamma\left(  -x+1,x-1,\right)
\right\vert .
\]
As we know already, for this it is enough to check that the cross-ratio
$R\left(  y-z,y+z;x-1,-x+1\right)  $ is small. Indeed, as we will check now,
\[
R\left(  y-z,y+z;x-1,-x+1\right)  <\left[  R\left(  -x-1,-x+1;x-1,x+1\right)
\right]  ^{1/2}\sim\frac{1}{x},
\]
so we will be done.

Let $y>0,$ then we will use only the relation $\frac{4z}{\left(  x-y\right)
^{2}-\left(  z+1\right)  ^{2}}\leq\frac{1}{x^{2}-1}.$ The worst case is when
\[
\frac{4z}{\left(  x-y\right)  ^{2}-\left(  z+1\right)  ^{2}}=\frac{1}{x^{2}%
-1}.
\]
It implies that
\begin{align}
y  &  =x-\sqrt{4z\left(  x^{2}-1\right)  +\left(  z+1\right)  ^{2}}%
,\label{81}\\
y^{2}  &  =x^{2}+4z\left(  x^{2}-1\right)  +\left(  z+1\right)  ^{2}%
-2x\sqrt{4z\left(  x^{2}-1\right)  +\left(  z+1\right)  ^{2}}\nonumber
\end{align}

We want to estimate the cross-ratio
\begin{align*}
& R\left(  y-z,y+z;x-1,-x+1\right)  \\
& =\frac{2z\left(  2x-2\right)  }{\left[  x-z-1+y\right]  \left[
x-z-1-y\right]  }=\frac{2z\left(  2x-2\right)  }{\left(  x-z-1\right)
^{2}-y^{2}}%
\end{align*}
By $\left(  \ref{81}\right)  $,
\begin{align*}
& R\left(  y-z,y+z;x-1,-x+1\right)  \\
& =\frac{2z\left(  2x-2\right)  }{-2x\left(  z+1\right)  +4z\left(
x^{2}-1\right)  +2x\sqrt{4z\left(  x^{2}-1\right)  +\left(  z+1\right)  ^{2}}%
}.
\end{align*}
Since $x$ is large, and $z<1$ we have
\begin{multline*}
\frac{2z\left(  2x-2\right)  }{-2x\left(  z+1\right)  +4z\left(
x^{2}-1\right)  +2x\sqrt{4z\left(  x^{2}-1\right)  +\left(  z+1\right)  ^{2}}%
}\\
\sim\frac{4z}{-2\left(  z+1\right)  +4zx+2\sqrt{4zx^{2}+1}}<\frac
{4z}{4z\left(  x-1\right)  }=\frac{1}{x-1},
\end{multline*}
which proves our claim.
\end{proof}

\section{Foliated states}

\subsection{Rigidity of a single interface \label{Section single}}

The Ising model on $\mathcal{L}_{p,q}$ is defined by the formal Hamiltonian%
\[
H\left(  \sigma\right)  =-\sum_{u\sim v}\sigma\left(  u\right)  \sigma\left(
v\right)  ,
\]
where $u,v\in\mathcal{L}_{p,q}$ are vertices of the graph $\mathcal{L}_{p,q},$
the function $\sigma\left(  \cdot\right)  $ takes values $\pm1,$ and the
summation goes over the nearest neighbours.

Let $\Gamma\pm$ be a geodesical family with phase assignment, see Construction
2 above. For every tesselation $\mathcal{L}_{p,q}$ this family defines a spin
configuration $\sigma_{\Gamma}^{\pm}$ on $\mathcal{L}_{p,q}$ in an evident way.

The rigidity of the interface $\Sigma_{\Gamma}$ of the low temperature Ising
model on $\mathcal{L}_{p,q},$ corresponding to the boundary condition
$\sigma_{\Gamma}^{\pm},$ in the case when $\Gamma$ consists of a single
geodesic $\gamma,$ is the main result of Series and Sinai, where the following
statement is proven:

\begin{theorem}
\label{single} (see \cite{SS}) Let $\gamma$ be a geodesics, $z\in\gamma$ be an
arbitrary point, and $\Sigma_{\gamma V}$ be the interface in the box $V$,
corresponding to the boundary condition $\sigma_{\gamma}^{\pm}$. Define the
neighborhood $\mathcal{C}_{m}\left(  \gamma,z\right)  $ of the curve $\gamma$
by
\begin{equation}
\mathcal{C}_{m}\left(  \gamma,z\right)  =\cup_{y\in\gamma}B\left(
y,r_{m}\left(  y\right)  \right)  ,\label{32}%
\end{equation}
where $B\left(  y,r\right)  \subset\mathcal{L}$ is a ball of radius $r,$
centered at $y,$ and the radius $r_{m}\left(  y\right)  $ is given by
\begin{equation}
r_{m}\left(  y\right)  =\max\left\{  m,\mathrm{dist}\left(  y,z\right)
\right\}  .\label{33}%
\end{equation}
then
\begin{equation}
\mathbb{P}_{V,\beta,\sigma_{\gamma}^{\pm}}\left\{  \Sigma_{\gamma
V}\not \subset \mathcal{C}_{m}\left(  \gamma,z\right)  \right\}  \leq
\exp\left\{  -\beta m\right\}  \label{39}%
\end{equation}
uniformly in $V.$
\end{theorem}

Here $\mathbb{P}_{V,\beta,\sigma_{\gamma}^{\pm}}$ is the Ising model Gibbs
state in the finite box $V\subset\mathcal{L}_{p,q},$ corresponding to the
inverse temperatire $\beta$ and boundary conditions $\sigma_{\gamma}^{\pm}.$

\subsection{Rigidity of finitely many interfaces}

When $\Gamma$ consists of finitely many geodesics $\gamma_{1},...,\gamma_{k}$,
a similar result holds, and the proof needs only one extra element, as
compared with the theorem \ref{single}, which element was already used for the
lemma \ref{family}. We start with some definitions.

Let our geodesical family $\Gamma$ be defined by the sequence $X_{2k}$ of
points on the absolute, $\left\{  x_{1}^{\prime},x_{1}^{\prime\prime
},...,x_{k}^{\prime},x_{k}^{\prime\prime}\right\}  ,$ going clockwise. In
other words, $\Gamma=\left\{  \gamma_{1}\left(  x_{1}^{\prime},x_{1}%
^{\prime\prime}\right)  ,...,\gamma_{k}\left(  x_{k}^{\prime},x_{k}%
^{\prime\prime}\right)  \right\}  $. The family $\Gamma$ defines the partition
$\Pi_{0}$ of $X_{2k}$ into $k$ pairs: $\left(  x_{1}^{\prime},x_{1}%
^{\prime\prime}\right)  ,...,\left(  x_{k}^{\prime},x_{k}^{\prime\prime
}\right)  .$ Every spin configuration $\sigma$ has its interface collection
$\Sigma_{\Gamma}\left(  \sigma\right)  ,$ and so defines a partition
$\Pi\left(  \sigma\right)  $ of $X_{2k}$ into pairs. The partition $\Pi_{0}$
should be called the ground state partition.

(Of course, the observable $\Pi\left(  \sigma\right)  ,$ as defined, is not
local. One has to talk about finite boxes $V$ with boundary condition
$\sigma_{\Gamma}^{\pm},$ and then the corresponding observable $\Pi_{V}\left(
\sigma\right)  $ is local, evidently. However, our estimates will be uniform
in $V,$ so we will talk about the observable $\Pi\left(  \sigma\right)  ,$
omitting the index $V.$)

\begin{theorem}
\label{finite}

\textbf{1. }Suppose the geodesical family $\left(  \Gamma,\pm\right)
=\left\{  \gamma_{1},...,\gamma_{k}\right\}  $ satisfies the conditions of
Lemma \ref{family} with $\alpha$ small. Then for every $V$ finite the
probability of the event $\Pi\left(  \sigma\right)  \neq\Pi_{0}$ satisfies%
\begin{equation}
\mathbb{P}_{V,\beta,\sigma_{\Gamma}^{\pm}}\left\{  \Pi\left(  \sigma\right)
\neq\Pi_{0}\right\}  \leq C\left(  k\right)  \exp\left\{  -\beta C\left(
\alpha\right)  \right\}  , \label{42}%
\end{equation}
where $C\left(  k\right)  \sim k^{2},$ while $C\left(  \alpha\right)  \sim
\ln\alpha^{-1}\rightarrow\infty$ as $\alpha\rightarrow0.$ In words, the
typical collection $\Sigma_{\Gamma}$ of $k$ interfaces pairs the points
$x_{1}^{\prime},x_{1}^{\prime\prime},...,x_{k}^{\prime},x_{k}^{\prime\prime}$
in the `correct' way: $\left(  x_{1}^{\prime},x_{1}^{\prime\prime}\right)
,...,\left(  x_{k}^{\prime},x_{k}^{\prime\prime}\right)  .$

\textbf{2. }Let $z_{i}\in\gamma_{i},$ $i=1,...,k$ be an arbitrary collection
of points on geodesics $\gamma_{i},$ and the sets $\mathcal{C}_{m}\left(
\gamma_{i},z_{i}\right)  $ are defined by $\left(  \ref{32},\ref{33}\right)
.$ Then under condition that $\Pi\left(  \sigma\right)  =\Pi_{0}$ we have%
\begin{equation}
\mathbb{P}_{V,\beta,\sigma_{\Gamma}^{\pm}}\left\{  \Sigma_{\Gamma}\left(
\sigma\right)  \not \subset \cup_{i=1}^{k}\left(  \mathcal{C}_{m}\left(
\gamma_{i},z_{i}\right)  \right)  {\LARGE |}\Pi\left(  \sigma\right)  =\Pi
_{0}\right\}  \leq k\exp\left\{  -\beta m\right\}  , \label{41}%
\end{equation}
uniformly in $V$ and the set $\left\{  z_{i},i=1,...,k\right\}  .$ In
particular, for the value $m\left(  \alpha\right)  =\frac{1}{2}\ln\alpha^{-1}$
we have
\[
\mathbb{P}_{V,\beta,\sigma_{\Gamma}^{\pm}}\left\{  \Sigma_{\Gamma}\left(
\sigma\right)  \not \subset \cup_{i=1}^{k}\left(  \mathcal{C}_{m\left(
\alpha\right)  }\left(  \gamma_{i},z_{i}\right)  \right)  \right\}  \leq
C\left(  k\right)  \exp\left\{  -\beta m\left(  \alpha\right)  \right\}  .
\]

\end{theorem}

\begin{proof}
The relation $\left(  \ref{41}\right)  $ follows in a straightforward way from
the Proposition 4.1 of \cite{SS}, even after replacing the partition $\Pi_{0}$
by any other allowed partition.

To see $\left(  \ref{42}\right)  ,$ let us start with the case $k=2.$ The
event we are interested in is that the partition $\Pi\left(  \sigma\right)  $
is `wrong': the two interfaces from $\Sigma_{\Gamma}\left(  \sigma\right)  $
connect $x_{1}^{\prime}$ to $x_{2}^{\prime\prime},$ and $x_{1}^{\prime\prime}$
-- to $x_{2}^{\prime}.$ Consider the `wrong' geodesics $\tilde{\gamma}%
_{1}\left(  x_{1}^{\prime},x_{2}^{\prime\prime}\right)  $ and $\tilde{\gamma
}_{2}\left(  x_{1}^{\prime\prime},x_{2}^{\prime}\right)  .$ As we know
already, the surplus -- or the difference -- $\left\vert \tilde{\gamma}%
_{1}\right\vert +\left\vert \tilde{\gamma}_{2}\right\vert -\left\vert
\gamma_{1}\right\vert -\left\vert \gamma_{2}\right\vert \sim2L\left(
\alpha\right)  =2\ln\alpha^{-1}$ is the minimal extra length the interface
$\Sigma_{\Gamma}\left(  \sigma\right)  $ has to pay for the wrong connection,
and this is the reason for $\left(  \ref{42}\right)  $ to hold. The argument
goes as follows.

Let $d\left(  \sigma\right)  $ be the distance between the two interfaces
$\eta_{1},\eta_{2}$, making $\Sigma_{\Gamma}\left(  \sigma\right)  .$ Suppose
first that $d\left(  \sigma\right)  \leq c_{1}L\left(  \alpha\right)  ,$ for
some suitable constant $c_{1}$ to be specified later. Let $\delta\left(
\sigma\right)  $ be the corresponding path, connecting $\eta_{1}$ to $\eta
_{2},$ $\left\vert \delta\right\vert \leq c_{1}L\left(  \alpha\right)  .$
Denote by $D\left(  \sigma\right)  =\max\left(  \mathrm{dist}\left(
\delta,\gamma_{1}\right)  ,\mathrm{dist}\left(  \delta,\gamma_{2}\right)
\right)  .$ Note that
\[
D\left(  \sigma\right)  \geq\frac{1}{2}\left(  1-c_{1}\right)  L\left(
\alpha\right)  .
\]
Consider the tubular neighborhood $\Delta\left(  \sigma\right)  $ of
$\delta\left(  \sigma\right)  $ of width $C_{3},$ and let us perform a Peierls
transformation of $\sigma\rightarrow\sigma^{\prime},$ flipping $\sigma$ inside
the contour $\partial=\partial\Delta\left(  \sigma\right)  .$ The result on
$\eta_{1},\eta_{2}$ of this surgery is a new pair $\eta_{1}^{\prime},\eta
_{2}^{\prime},$ connecting now $x_{1}^{\prime}$ to $x_{1}^{\prime\prime}$, and
$x_{2}^{\prime}$ -- to $x_{2}^{\prime\prime},$ so $\Pi\left(  \sigma^{\prime
}\right)  =\Pi_{0}.$ Note that the Hausdorff distance $\mathrm{dist}_{H}$
satisfies%
\[
\max_{i=1,2}\mathrm{dist}_{H}\left(  \eta_{i}^{\prime},\gamma_{i}\right)  \geq
D\left(  \sigma\right)  -C_{3},
\]
thus the event $\Sigma_{\Gamma}\left(  \sigma^{\prime}\right)  \not \subset
\cup_{i=1}^{2}\left(  \mathcal{C}_{m}\left(  \gamma_{i},z_{i}\right)  \right)
$ happens, with $m=D\left(  \sigma\right)  -C_{3}.$ Here the points $z_{i}%
\in\gamma_{i}$ are chosen in such a way that $\mathrm{dist}\left(  z_{1}%
,z_{2}\right)  =\mathrm{dist}\left(  \gamma_{1},\gamma_{2}\right)  .$ On the
other hand, $H\left(  \sigma^{\prime}\right)  -H\left(  \sigma\right)  \leq
C\left(  m,q\right)  \left\vert \delta\left(  \sigma\right)  \right\vert \leq
C\left(  m,q\right)  c_{1}L\left(  \alpha\right)  $ for some constant
$C\left(  m,q\right)  ,$ which arises due to the difference of the metrics on
$\mathcal{L}$ and $\mathcal{L}_{p,q}.$ The usual Peierls argument computations
together with $\left(  \ref{41}\right)  $ imply that
\begin{align}
&  \mathbb{P}_{V,\beta,\sigma_{\Gamma}^{\pm}}\left\{  \sigma:\Pi\left(
\sigma\right)  \neq\Pi_{0},d\left(  \sigma\right)  \leq c_{1}L\left(
\alpha\right)  \right\} \nonumber\\
&  \leq C_{4}\exp\left\{  \beta C\left(  p,q\right)  c_{1}L\left(
\alpha\right)  \right\}  \exp\left\{  -\beta\left(  D\left(  \sigma\right)
-C_{3}\right)  \right\} \label{44}\\
&  =C_{4}\exp\left\{  -\beta\left(  \frac{1}{2}\left(  1-c_{1}\right)
L\left(  \alpha\right)  -C_{3}-C\left(  p,q\right)  c_{1}L\left(
\alpha\right)  \right)  \right\}  .
\end{align}

In the opposite case, when $d\left(  \sigma\right)  >c_{1}L\left(
\alpha\right)  ,$ we note that necessarily $\Sigma_{\Gamma}\left(
\sigma\right)  \not \subset \cup_{i=1}^{2}\left(  \mathcal{C}_{m}\left(
\tilde{\gamma}_{i},\tilde{z}_{i}\right)  \right)  $ for $\tilde{z}_{1}%
,\tilde{z}_{2}$ defined by $\mathrm{dist}\left(  \tilde{z}_{1},\tilde{z}%
_{2}\right)  =\mathrm{dist}\left(  \tilde{\gamma}_{1},\tilde{\gamma}%
_{2}\right)  ,$and $m=\frac{1}{2}c_{1}L\left(  \alpha\right)  .$ (Here we use
the obvious fact that $\mathrm{dist}\left(  \tilde{\gamma}_{1},\tilde{\gamma
}_{2}\right)  $ goes to $0$ as $\alpha\rightarrow0.)$ Therefore, again by
$\left(  \ref{41}\right)  ,$
\begin{equation}
\mathbb{P}_{V,\beta,\sigma_{\Gamma}^{\pm}}\left\{  \sigma:\Pi\left(
\sigma\right)  \neq\Pi_{0},d\left(  \sigma\right)  \leq c_{1}L\left(
\alpha\right)  \right\}  \leq2\exp\left\{  -\frac{1}{2}\beta c_{1}L\left(
\alpha\right)  \right\}  . \label{45}%
\end{equation}
Comparing $\left(  \ref{44}\right)  $ with $\left(  \ref{45}\right)  ,$ we see
that the optimal choice of $c_{1}$ is given by
\[
c_{1}=\frac{1}{2+C\left(  p,q\right)  },
\]
so our claim follows. (It is easy to see, by the way, that $C\left(
p,q\right)  $ is bounded by a universal constant.)

The case of $k>2$ is treated in the same way as above, using the computations
made in the proof of the Lemma \ref{family}.
\end{proof}

\subsection{Rigidity of the \textbf{Foliated State interfaces}}

As we saw in the previous subsection, in the low-temperature state defined by
the boundary condition $\sigma_{\Gamma}^{\pm},$ corresponding to some finite
geodesical family $\left(  \Gamma,\pm\right)  =\left\{  \gamma_{1}%
,...,\gamma_{k}\right\}  ,$ we have $\Pi\left(  \sigma\right)  =\Pi_{0}$ for a
typical configuration $\sigma.$ However, the temperature for which this claim
is true, goes to $0$ as $k\rightarrow\infty.$ Moreover, in the case of
infinitely many interfaces, corresponding to the millefeuille family
$\Gamma_{\infty}\left(  \alpha,\eta\right)  $ and the boundary condition
$\sigma_{\Gamma_{\infty}}^{\pm}$ the probability $\mathbb{P}_{V,\beta
,\sigma_{\Gamma_{\infty}}^{\pm}}\left\{  \Pi\left(  \sigma\right)  =\Pi
_{0}\right\}  $ goes to zero as $V\rightarrow\infty$ for every finite
temperature. Indeed, for every geodesic $\gamma\in\Gamma_{\infty}$ there are
infinitely many curves $\tilde{\gamma}$ in $\Gamma_{\infty},$ for which
$\mathrm{dist}\left(  \gamma,\tilde{\gamma}\right)  \leq\max\left\{  L\left(
\alpha\right)  ,L\left(  \eta\right)  \right\}  .$ Therefore, the probability
in the state $\mathbb{P}_{\beta,\sigma_{\Gamma_{\infty}}^{\pm}}$that somewhere
along $\gamma$ the surgery between it and $\tilde{\gamma}$ will happen, equals
1 for any positive temperature $\beta^{-1}.$ Nevertheless, the density of
these surgeries, for $\alpha$ small, goes to zero as $\beta\rightarrow\infty,$
which means that in the vicinity of any fixed point the probability that such
a surgery happens there is small (but it does depend on the size of the
neighborhood -- as is also the case for the rigidity property of the Dobrushin
interface in $\mathbb{Z}^{3}$).

Now we will formulate one version of the theorem which makes this claim
rigorous. Let $\gamma_{1},\gamma_{2}\in\Gamma_{\infty}\left(  \alpha
,\eta\right)  $ be the first two geodesics of the family, say, and let the
point $Z\in\mathcal{L}$ be defined by $\mathrm{dist}\left(  Z,\gamma
_{1}\right)  =\mathrm{dist}\left(  Z,\gamma_{2}\right)  =\frac{1}%
{2}\mathrm{dist}\left(  \gamma_{1},\gamma_{2}\right)  .$ In other words, the
point $Z$ is the symmetry center of the pair $\gamma_{1},\gamma_{2}.$ As
before, we denote by $\sigma_{\Gamma_{\infty}}^{\pm}$ some ground state
configuration, corresponding to the geodesical family $\Gamma_{\infty}\left(
\alpha,\eta\right)  $. Without loss of generality we can assume that
$\sigma_{\Gamma_{\infty}}^{\pm}\left(  t\right)  =+1$ for sites $t\in
U_{r}\left(  Z\right)  \subset\mathcal{L}_{p,q}$ which are at distance $\leq
r$ from the point $Z\in\mathcal{L};$ here $r$ is some finite radius, much
smaller than our scale $\ln\alpha^{-1}.$ Let $V\subset\mathcal{L}_{p,q}$ be a
finite box with boundary condition $\sigma_{\Gamma_{\infty}}^{\pm}%
{\LARGE |}_{V^{c}}$ defined by $\sigma_{\Gamma_{\infty}}^{\pm},$ and
$\sigma_{V}\in\Omega_{V}$ be a spin configuration in $V.$ Let $\Sigma\left(
\sigma_{\Gamma_{\infty}}^{\pm}\right)  $ be the (countable) collection of all
(infinite) contours of $\sigma_{\Gamma_{\infty}}^{\pm},$ while $\Sigma\left(
\sigma_{\Gamma_{\infty}}^{\pm}{\LARGE |}_{V^{c}}\cup\sigma_{V}\right)  $ --
the collection of all infinite contours of $\sigma_{\Gamma_{\infty}}^{\pm
}{\LARGE |}_{V^{c}}\cup\sigma_{V}.$ Denote by $\Lambda\left(  \sigma
_{V}\right)  $ the symmetric difference $\Sigma\left(  \sigma_{\Gamma_{\infty
}}^{\pm}\right)  \bigtriangleup\Sigma\left(  \sigma_{\Gamma_{\infty}}^{\pm
}{\LARGE |}_{V^{c}}\cup\sigma_{V}\right)  ;$ it consists of finitely many
closed contours.

\begin{theorem}
Let the parameters $\alpha$ and $\eta$ are small enough, and the temperature
$\beta^{-1}$ is low. Then typically the interfaces $\Sigma\left(
\sigma_{\Gamma_{\infty}}^{\pm}{\LARGE |}_{V^{c}}\cup\sigma_{V}\right)  $ are
far away from the point $Z,$ and the phase $\left\langle \cdot\right\rangle
_{\beta,\sigma_{\Gamma_{\infty}}^{\pm}}$ is $e^{-\beta\ln\alpha^{-1}}$-close
to the phase $\left\langle \cdot\right\rangle _{+},$ when restricted to the
box $U_{r}\left(  Z\right)  .$ Explicitly, for some $C$ and $C\left(
r\right)  $%
\[
\mathbb{P}_{V,\beta,\sigma_{\Gamma_{\infty}}^{\pm}}\left(  \Lambda\left(
\sigma_{V}\right)  \cap U_{r}\left(  Z\right)  \neq\varnothing\right)  \leq
C\left(  r\right)  \exp\left\{  -C\beta\ln\alpha^{-1}\right\}
\]
uniformly in $V.$ In words, the interfaces typically are far away from
$U_{r}\left(  Z\right)  .$
\end{theorem}

\begin{proof}
Let the event $\Lambda\left(  \sigma_{V}\right)  \cap U_{r}\left(  Z\right)
\neq\varnothing$ does happen, and $\theta\equiv\theta\left(  \sigma
_{V},Z\right)  \in\Lambda\left(  \sigma_{V}\right)  $ be a contour which
contributes to the event $\Lambda\left(  \sigma_{V}\right)  \cap U_{r}\left(
Z\right)  \neq\varnothing.$ Then the loop $\theta$ is made by fragments of
several interfaces from the family $\Sigma\left(  \sigma_{\Gamma_{\infty}%
}^{\pm}\right)  $ of contours of $\sigma_{\Gamma_{\infty}}^{\pm},$ and by the
same number of interfaces from $\Sigma\left(  \sigma_{\Gamma_{\infty}}^{\pm
}{\LARGE |}_{V^{c}}\cup\sigma_{V}\right)  .$ Denote this observable by
$k\left(  \theta\left(  \sigma_{V},Z\right)  \right)  ,$ and extend it to all
configurations $\sigma_{V}$ by defining $k\left(  \theta\left(  \sigma
_{V},Z\right)  \right)  =0$ if $\Lambda\left(  \sigma_{V}\right)  \cap
U_{r}\left(  Z\right)  =\varnothing.$ We will prove our theorem by induction
on $k$, estimating the probabilities of the events $\mathbb{P}_{V,\beta
,\sigma_{\Gamma_{\infty}}^{\pm}}\left(  k\left(  \theta\left(  \sigma
_{V},Z\right)  \right)  =k\right)  $ for each $k\geq1.$

\textbf{Case }$k=1.$ That means that for one of the geodesics $\gamma_{i}%
\in\Gamma_{\infty}$ the event $\Sigma_{\gamma_{i}V}\not \subset
\mathcal{C}_{m_{i}}\left(  \gamma_{i},z_{i}\right)  $ happens, where $z_{i}%
\in\gamma_{i}$ is chosen to be the point on $\gamma_{i}$ closest to $Z,$ and
$m_{i}=\mathrm{dist}\left(  Z,\gamma_{i}\right)  .$ We know already, that
\begin{equation}
\mathbb{P}_{V,\beta,\sigma_{\Gamma_{\infty}}^{\pm}}\left\{  \Sigma_{\gamma
_{i}V}\not \subset \mathcal{C}_{m_{i}}\left(  \gamma_{i},z_{i}\right)
\right\}  \leq\exp\left\{  -\beta m_{i}\right\}  . \label{38}%
\end{equation}
So we need to know how fast the sequence $m_{i}$ grows.

\begin{figure}[h]
\centering
\includegraphics[width=13cm]{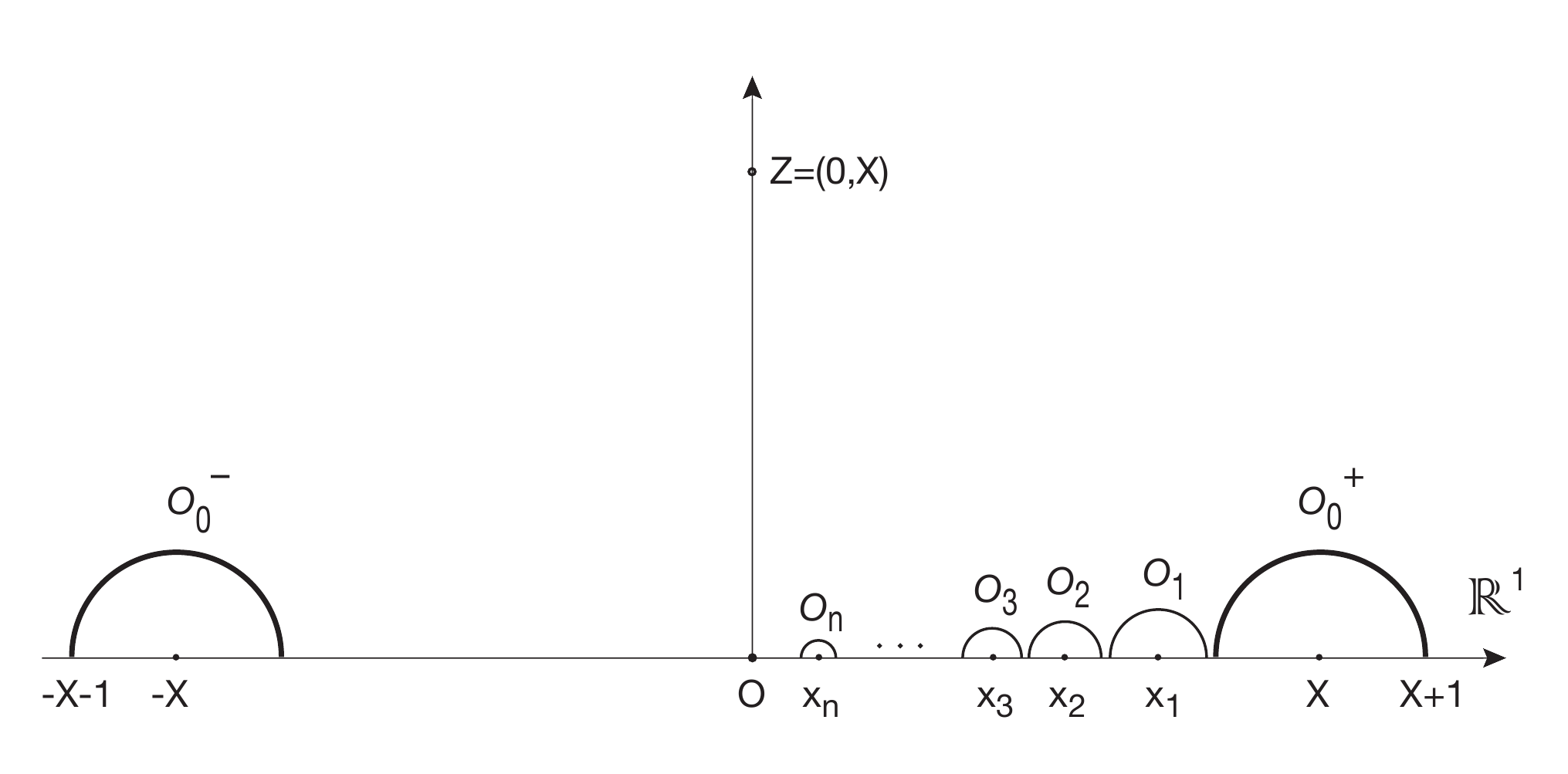}\caption{The geodesics
$\gamma_{i}$ in the half--plane model.}%
\label{Fol_States}%
\end{figure}

The computations are easier in the upper half--plane model. Without loss of
generality we can assume that the first two geodesics $\gamma_{1},\gamma_{2}$
are the two semicircles, $O_{0}^{-},O_{0}^{+}$ in the upper half--plane,
centered at points $\left(  -X,0\right)  $ and $\left(  X,0\right)  ,$ with
radius $1.$ $X$ is related to $\alpha$ by $X\sim\ln\alpha^{-1}.$ Let us take
for the point $Z$ the point $\left(  0,X\right)  .$ It is sufficient to
consider only those geodesics $\gamma_{i}$ -- i.e. semicircles $O_{i}$ --
which are located between the $y$-axis and the semicircle $O_{0}^{+}.$ This is
only `a quarter' of all geodesics from $\Gamma_{\infty},$ but it is sufficient
for our question. So let $X>x_{1}>x_{2}>...>x_{n}>0$, $r_{1},...,r_{n}>0$ be a
sequence of centers and a sequence of radii of non-intersecting semicircles
$O_{i}$, see Fig.\ref{Fol_States}.

It turns out that the divergence we need does not require the precise
information about the structure of the family $\left\{  O_{i}\right\}  ;$ in
particular, we will not use in the essential way the fact that the
cross-ratios for certain pairs $O_{i},O_{j}$ are our small numbers $\alpha$
and $\eta.$ The only thing needed is that the semicircles $O_{i}$ do not
intersect, and that all the radii $r_{i}\leq\frac{1}{2}$. We want to estimate
the distances $\rho_{i}=\mathrm{dist}\left(  Z,O_{i}\right)  ;$ we need them
to diverge to infinity fast enough. It is sufficient to estimate the distances
$\varrho_{i}=\mathrm{dist}\left(  O_{0}^{-},O_{i}\right)  ,$ since
$\mathrm{dist}\left(  o,O_{0}^{-}\right)  $ is just a constant. The latter is
simpler, since the distances $\varrho_{i}$ can be expressed via the
cross-ratios:%
\[
\varrho_{i}\cong-\ln\frac{\left(  2\right)  \left(  2r_{i}\right)  }{\left(
X+1+x_{i}+r_{i}\right)  \left(  X-1+x_{i}-r_{i}\right)  }\gtrsim-\ln
\frac{r_{i}}{X^{2}}.
\]
According to $(\ref{38})$, the probability in question is
bounded by
\[
\sum_{i=1}^{n}\exp\left\{  -\beta\varrho_{i}\right\}  \lesssim\sum_{i=1}%
^{n}\exp\left\{  \beta\ln\frac{r_{i}}{X^{2}}\right\}  =\sum_{i=1}^{n}\left(
\frac{r_{i}}{X^{2}}\right)  ^{\beta}=X^{-2\beta}\sum_{i=1}^{n}\left(
r_{i}\right)  ^{\beta}.
\]
Note that $2\left(  r_{1}+...+r_{n}\right)  \leq X,$ and that $r^{\prime\beta
}+r^{\prime\prime\beta}<\left(  r^{\prime}+r^{\prime\prime}\right)  ^{\beta}$
for $\beta>1.$ Therefore, for any $n$ we have
\[
\sum_{i=1}^{n}\left(  r_{i}\right)  ^{\beta}\leq\sum_{i=1}^{\left[  X\right]
}1=X,
\]
so
\[
\sum_{i=1}^{n}\exp\left\{  -\beta\varrho_{i}\right\}  \leq X^{-2\beta+1}.
\]

\textbf{Case }$k=2.$ The events $\mathbb{P}_{V,\beta,\sigma_{\Gamma_{\infty}%
}^{\pm}}\left(  k\left(  \theta\left(  \sigma_{V},Z\right)  \right)
=2\right)  $ means that there are two geodesics $\gamma^{\prime}\left(
x_{1},x_{2}\right)  ,\gamma^{\prime\prime}\left(  x_{3},x_{4}\right)
\in\Gamma_{\infty},$ such that the configuration $\sigma_{\Gamma_{\infty}%
}^{\pm}{\LARGE |}_{V^{c}}\cup\sigma_{V}$ has, among other, two interfaces
$\Sigma\left(  x_{1},x_{4}\right)  $ and $\Sigma\left(  x_{2},x_{3}\right)  ,$
so that the \textquotedblleft contour $\left[  \Sigma\left(  x_{1}%
,x_{4}\right)  \cup\Sigma\left(  x_{2},x_{3}\right)  \right]  \bigtriangleup
\left[  \gamma^{\prime}\left(  x_{1},x_{2}\right)  \cup\gamma^{\prime\prime
}\left(  x_{3},x_{4}\right)  \right]  $\textquotedblright\ surrounds the point
$Z.$ As is established during the proof of the Theorem \ref{finite}, the
probability of such an event is of the order of $\exp\left\{  -\beta\left(
\mathrm{dist}\left(  Z,\gamma^{\prime}\right)  +\mathrm{dist}\left(
Z,\gamma^{\prime\prime}\right)  \right)  \right\}  .$ In the previous
paragraph we have shown that
\[
\sum_{\gamma^{\prime},\gamma^{\prime\prime}\in\Gamma_{\infty}}\exp\left\{
-\beta\left(  \mathrm{dist}\left(  Z,\gamma^{\prime}\right)  +\mathrm{dist}%
\left(  Z,\gamma^{\prime\prime}\right)  \right)  \right\}  \leq\left[  \left(
\ln\alpha^{-1}\right)  ^{-2\beta+1}\right]  ^{2},
\]
and we are done.

The rest of the proof goes by induction on $k.$
\end{proof}

\section{Periodic Gibbs states on Cayley trees $\mathcal{T}_{n}$ for
$n=2,3$}

Motivated by the paper \cite{RR}, we have constructed in \cite{GRS} a huge
manifold of extremal Gibbs states on Cayley trees $\mathcal{T}_{n}.$ However,
our construction, as well as that in \cite{RR}, was restricted to the case
$n\geq4;$ we were using the fact that each vertex of our tree has at least $5$
neighbors. However, that, happily, does not mean that the trees $\mathcal{T}%
_{2}$ and $\mathcal{T}_{3}$ are that different from the rest. The difference
between $\mathcal{T}_{2},\mathcal{T}_{3}$ and $\mathcal{T}_{n\geq4}$ lies only
in the fact that the former trees do not carry states with period two. Thus,
the equations of \cite{RR} do not have solutions for $n=2,3,$ since they are
period two functions, while the equations for functions with bigger periods
are too complicated to be analyzed. On the other hand, our method to construct
`dimer' states and their analogs, needs just a tiny modification to be
applicable to all $n\geq2.$ This modification is explained below. To simplify
the exposition we will consider only the case of $\mathcal{T}_{2}.$ We start
with a definition.

\begin{definition}
A $k$-chain $C\subset E$ is a sequence $x_{0},x_{1},...,x_{k}\in V$ of
distinct n.n. vertices of $\mathcal{T}_{2}\equiv\left(  V,E\right)  .$ (For
example, a dimer is a $1$-chain.) The vertices $x_{0}$ and $x_{k}$ are called
the ends of the chain.
\end{definition}

The set of all $k$-chains will be denoted by $\mathcal{C}_{k}.$

A collection $R_{k}\subset\mathcal{C}_{k}$ of $k$-chains will be called a
\textit{covering} of $\mathcal{T}_{2}$, if every vertex $x$ of $\mathcal{T}%
_{2}$ belongs to precisely one $k$-chain from $R_{k}.$ The existence of
coverings is evident.

Let $R_{k}$ be such a covering, and suppose that $k$ is odd, $k=2m+1$. Then
every $k$-chain has a \textit{middle} dimer. The collection $D=\left\{
d_{i}=\left(  y_{i},y_{i}^{\prime}\right)  \in E\right\}  $ of these dimers is
called a $k$-covering, associated with $R_{k},$ $D=D\left(  R_{k}\right)  .$
(Of course, a $k$-covering is not a covering for $k>1.$)

Let a $k$-covering $D$ be fixed. Consider the Ising spin configuration
$\sigma_{D}$ on $\mathcal{T}_{2},$ defined by the property: for every two n.n.
sites $z,z^{\prime}\in V$%
\[
\sigma_{D}\left(  z\right)  \sigma_{D}\left(  z^{\prime}\right)  =\left\{
\begin{array}
[c]{cc}%
-1 & \text{ for }\left(  z,z^{\prime}\right)  \in D\text{ }\\
+1 & \text{ for }\left(  z,z^{\prime}\right)  \notin D
\end{array}
\right.  .
\]
In fact, there are exactly two such configurations, which differ by a global
spin-flip. We choose one, see Fig.\ref{3-covering}. Our main observation is that

\begin{claim}
For $k\geq 5$  the configuration $\sigma _{D}$ is a ground state configuration.
 Moreover, it is a stable ground state, which means that there
exist a family of low temperature Ising model Gibbs states $\left\langle
\cdot\right\rangle _{\beta}$ on $\mathcal{T}_{2},$ which converges to
$\sigma_{D}$ weakly, as $\beta\rightarrow\infty.$ In fact, for trees
$\mathcal{T}_{n\geq4}$ this is true even for $k=1,$ see \cite{GRS}, but for
$\mathcal{T}_{2}$ this is not the case.
\end{claim}

\begin{figure}[h]
\centering
\includegraphics[width=11cm]{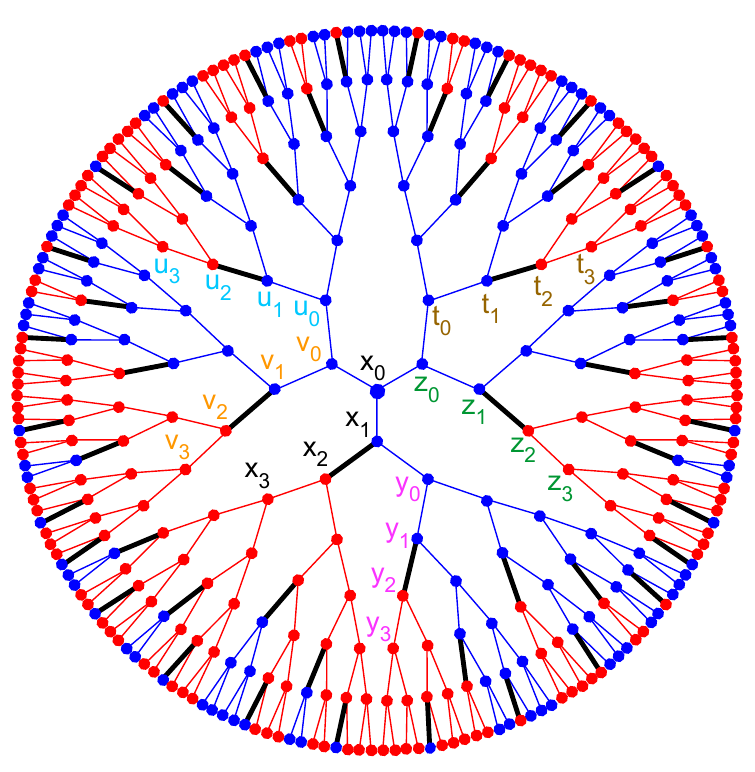}\caption{Covering $R_{3}$ of
$\mathcal{T}_{2}$, 3--covering $D(R_3)$ (black bonds)
and  spin configuration $\sigma_{D}$. Blue and red sites have opposite
signs. The first six chains are marked.}%
\label{3-covering}%
\end{figure}

To see this, we put $\mathcal{T}_{2}$ on $\mathbb{R}^{2},$ since all Cayley
trees are planar, and we will talk about contours, which are closed loops on
$\mathbb{R}^{2},$ which intersect the bonds of our tree, but which do not pass
through the vertices.

For every loop $\gamma,$ which is a Peierls-like contour, we define the length
$\left\vert \gamma\right\vert $ of $\gamma$ to be the number of bonds of
$\mathcal{T}_{2}$ that $\gamma$ traverses, and let $\left\vert \gamma
\right\vert _{D}\leq\left\vert \gamma\right\vert $ be the number of bonds
among them which are from $D.$

Consider the ratio $\frac{\left\vert \gamma\right\vert _{D}}{\left\vert
\gamma\right\vert },$ and let%
\[
\varphi\left(  D\right)  =\sup_{\gamma}\frac{\left\vert \gamma\right\vert
_{D}}{\left\vert \gamma\right\vert },
\]
where the $\sup$ is taken over all finite contours. Our claim follows
immediately from the following Peierls stability property:

\begin{lemma}%
\begin{equation}
\varphi\left(  D\right)  \leq\frac{1}{m+1}.\label{46}%
\end{equation}

\end{lemma}

Note that for paths $\gamma$ which are not loops this is not true. Moreover,
one can find paths $\gamma$ -- even long paths -- for which $\left\vert
\gamma\right\vert _{D}=\left\vert \gamma\right\vert .$

We will have Peierls condition satisfied, if for some (small) $c>0$ and any
$\gamma$ we have $\left\vert \gamma\right\vert -2\left\vert \gamma\right\vert
_{D}>c\left\vert \gamma\right\vert ,$ which means that we need $\left\vert
\gamma\right\vert _{D}<\frac{1-c}{2}\left\vert \gamma\right\vert .$ According
to $\left(\ref{46}\right)$, the desired relation holds once $\frac{1}%
{m+1}<\frac{1}{2},$ i.e. when $m\geq2$ and $k\geq5.$

\begin{proof}
Let $\mathcal{T}_{2}$ be embedded into upper half--plane $\mathbb{R}^{2+}.$
Choose an arbitrary vertex $\mathbf{0}\in\mathcal{T}_{2}$, which will be
called the root of the tree. Let it be the only vertex with $y$-coordinate
equals to zero. We suppose that all the bonds are of the length one. First we
will construct a concrete covering $R$ of $\mathcal{T}_{2}$ by $k$-chains. The
construction is by induction. The first chain starts from $\mathbf{0}%
\in\mathcal{T}_{2},$ and each of its bonds is the left one among the two (or
three in the first step) possible options. Suppose the $k$-chains $C_{i}$ are
already defined, $i\leq n-1.$ Let the vertex $x\in\mathcal{T}_{2}$ be the
closest to $\mathbf{0\ }$among those not yet covered by all the $C_{i}.$ If
there are several such vertices, we take the leftmost one. The chain $C_{n}$
starts at $x$ and then always uses the `left' bonds.

Let $\gamma\subset\mathbb{R}^{2}$ be a loop, surrounding $\mathbf{0,}$ and
such that $\gamma\cap V=\varnothing.$ We will prove the inequality
\[
\frac{\left\vert \gamma\right\vert _{D}}{\left\vert \gamma\right\vert }%
\leq\frac{1}{m+1}%
\]
by induction in the number $\left\vert \mathrm{Int}\left(  \gamma\right)  \cap
V\right\vert $ of the vertices of $\mathcal{T}_{2},$ surrounded by $\gamma.$
Note that if $\left\vert \mathrm{Int}\left(  \gamma\right)  \cap V\right\vert
\leq m,$ then $\left\vert \gamma\right\vert _{D}=0,$ so the initial step of
induction is done.

For any $y\neq\mathbf{0}\in V$ denote by $l_{y}$ and $r_{y}$ the left and the
right bonds, starting at $y$ and going away from $\mathbf{0},$ and by $b_{y}$
the bond going `back' to $\mathbf{0}$.

Suppose the lemma is true for all $\gamma$ with $\left\vert \mathrm{Int}%
\left(  \gamma\right)  \cap V\right\vert \leq N-1.$ Take some loop $\gamma$
with $\left\vert \mathrm{Int}\left(  \gamma\right)  \cap V\right\vert =N.$ Let
$x\in\mathrm{Int}\left(  \gamma\right)  $ be the point in $\mathrm{Int}\left(
\gamma\right)  $ with maximal distance from $\mathbf{0.}$ Because of
maximality of $x$, $\gamma$ intersects both $l_{x}$ and $r_{x}.$

Suppose first that $l_{x}\notin D.$ (By construction, $r_{x}\notin D$). Let us
deform $\gamma$ in the vicinity of $x$ into $\gamma^{\prime}$, moving $\gamma$
in such a way that instead of intersecting $l_{x}$ and $r_{x}$ it intersects
just $b_{x}.$ The situation for all other bonds stays the same. In other
words, $\mathrm{Int}\left(  \gamma^{\prime}\right)  =\mathrm{Int}\left(
\gamma\right)  \smallsetminus x.$ Then $\left\vert \gamma^{\prime}\right\vert
=\left\vert \gamma\right\vert -1,$ while $\left\vert \gamma^{\prime
}\right\vert _{D}=\left\vert \gamma\right\vert _{D},$ so $\frac{\left\vert
\gamma\right\vert _{D}}{\left\vert \gamma\right\vert }<$ $\frac{\left\vert
\gamma^{\prime}\right\vert _{D}}{\left\vert \gamma^{\prime}\right\vert }.$ By
induction, $\frac{\left\vert \gamma^{\prime}\right\vert _{D}}{\left\vert
\gamma^{\prime}\right\vert }\leq\frac{1}{m+1},$ and we are done.

Now consider the case when $l_{x}\in D.$ Let $x_{0},x_{1},...,x_{m}=x$ be the
`lower' half of the $\left(  2m+1\right)  $-chain $C^{x}$, to which $x$
belongs. Note, that, by construction, none of the bonds $r_{x_{i}},$
$i=0,1,...,m$ belong to the $\left(  2m+1\right)  $-chains from our covering
$R.$ Let $T_{0},T_{1},...,T_{m}\subset\mathcal{T}_{2}$ be the subtrees,
defined as follows: $T_{i}$ is the maximal connected component of the
complement $\mathcal{T}_{2}\smallsetminus C^{x},$ containing the bond
$r_{x_{i}}.$ Since $\mathbf{0}\in\mathrm{Int}\left(  \gamma\right)  ,$ the
loop $\gamma$ has to intersect every tree $T_{i}$ at least once. However, due
to the maximality property of $x$, the intersection $\left[  \gamma\cap\left(
T_{0}\cup T_{1}\cup...\cup T_{m}\right)  \right]  $ is disjoint with $D;$ in
other words, $\left[  \gamma\cap\left(  T_{0}\cup T_{1}\cup...\cup
T_{m}\right)  \right]  \cap D=\varnothing.$ On the other hand, this
intersection $\left[  \gamma\cap\left(  T_{0}\cup T_{1}\cup...\cup
T_{m}\right)  \right]  $ involves at least $m+1$ bonds -- at least one bond
per any subtree $T_{i}.$ Let $T_{x_{m}}\subset\mathcal{T}_{2}$ be the subtree,
rooted at $x_{m}$ and consisting of all its descendants. Let us deform
$\gamma$ into $\gamma^{\prime\prime},$ where the latter loop is defined by two conditions:

1. \noindent$\mathrm{Int}\left(  \gamma^{\prime\prime}\right)  \cap T_{x_{m}%
}=\varnothing;$

2. \noindent$\mathrm{Int}\left(  \gamma^{\prime\prime}\right)  \cap\left[
\mathcal{T}_{2}\smallsetminus T_{x_{m}}\right]  =\mathrm{Int}\left(
\gamma\right)  \cap\left[  \mathcal{T}_{2}\smallsetminus T_{x_{m}}\right]  .$

One sees immediately that the bonds that contribute to $\left\vert
\gamma^{\prime\prime}\right\vert $ are these bonds in $\mathcal{T}%
_{2}\smallsetminus T_{x_{m}},$ which contribute to $\left\vert \gamma
\right\vert ,$ plus the bond $r_{x_{m}}.$ So $\left\vert \gamma^{\prime\prime
}\right\vert \leq\left\vert \gamma\right\vert -\left(  m+1\right)  .$ By
construction, $\left\vert \gamma^{\prime\prime}\right\vert _{D}=\left\vert
\gamma\right\vert _{D}-1.$ By induction, we have $\frac{\left\vert
\gamma\right\vert _{D}-1}{\left\vert \gamma\right\vert -\left(  m+1\right)
}\leq\frac{1}{m+1},$ which implies that $\frac{\left\vert \gamma\right\vert
_{D}}{\left\vert \gamma\right\vert }\leq\frac{1}{m+1}.$
\end{proof}

\subsection{Periodic states with non-zero magnetization}

The previous construction results in a ground state with zero mean
magnetization. We present now a generalization, which will have a non-zero
magnetization. It is analogous to our construction of such states on the
Lobachevsky plane.

Let $R$ be the covering by $k$-chains, constructed in the proof above, $D$ is
the associated $k$-covering by the dimers, and $\sigma_{D}$ is the
corresponding ground state. For every $k$-chain $C=\left\{  x_{0}%
,x_{1},...,x_{k}\right\}  \in R$ we call the site $x_{0}$ a $\left(  +\right)
$-end iff $\sigma_{D}\left(  x_{0}\right)  =+1.$ We denote it by $e_{+}\left(
C\right)  .$ Then the opposite end-point $x_{k}$ is, naturally, called a
$\left(  -\right)  $-end, and denoted by $e_{-}\left(  C\right)  .$ Let
$k=l+n-1,$ $l>n>0$. Define the collection of dimers $D_{l:n}$ as follows: on
every chain $C\in R$ let us take the bond $d\left(  C\right)  ,$ which is at
distance $l$ from $e_{+}\left(  C\right)  $ (and so at distance $n$ from
$e_{-}\left(  C\right)  $).

The collection $D_{l:n}$ satisfies%
\[
\varphi\left(  D_{l:n}\right)  \leq\frac{1}{n+1},
\]
as the Lemma above shows, so the configuration $\sigma_{D_{l:n}}$ is a stable
ground state, once $n$ is big enough. The nice property of it is that, unlike
the collection $D\equiv D_{\left(  m+1\right)  :\left(  m+1\right)  }$,
constructed earlier, its mean magnetization $m\left(  \sigma_{D_{l:n}}\right)
$ is non-zero.

\bigskip

\textit{Acknowledgement. We would like to thank M. Aizenman, P. Bleher, A. van Enter and S. Pirogov for
valuable discussions and comments.}

\vspace{2truecm}

\noindent \textbf{Daniel Gandolfo}\\
\noindent Aix Marseille Universit\'{e}, CNRS, CPT, UMR 7332, 13288 Marseille, France.
Universit\'{e} de Toulon, CNRS, CPT, UMR 7332, 83957 La Garde, France.\\
\noindent gandolfo@cpt.univ-mrs.fr, gandolfo@univ-tln.fr

\bigskip

\noindent \textbf{Jean Ruiz}\\
\noindent Aix Marseille Universit\'{e}, CNRS, CPT, UMR 7332, 13288 Marseille, France.
Universit\'{e} de Toulon, CNRS, CPT, UMR 7332, 83957 La Garde, France.\\
\noindent ruiz@cpt.univ-mrs.fr

\bigskip

\noindent  \textbf{Senya Shlosman} \\
\noindent Aix Marseille Universit\'{e}, CNRS, CPT, UMR 7332, 13288 Marseille, France.
Universit\'{e} de Toulon, CNRS, CPT, UMR 7332, 83957 La Garde, France.\\
\noindent shlosman@cpt.univ-mrs.fr, shlosman@univ-amu.fr\\
\noindent Inst. of the Information Transmission Problems, RAS, Moscow, Russia.\\
shlos@iitp.ru

\end{document}